\documentclass[12pt,twoside]{article}
\usepackage{latexsym,amsmath,amssymb}
\usepackage{wrapfig}
\usepackage{pgfplots}
\usepackage{pgfplotstable}
\usepgfplotslibrary{external}
\pgfplotsset{compat = newest}
\def\,{\mskip 3mu} \def\>{\mskip 4mu plus 2mu minus 4mu} \def\;{\mskip 5mu plus 5mu} \def\!{\mskip-3mu}
\def\dispmuskip{\thinmuskip= 3mu plus 0mu minus 2mu \medmuskip=  4mu plus 2mu minus 2mu \thickmuskip=5mu plus 5mu minus 2mu}
\def\textmuskip{\thinmuskip= 0mu                    \medmuskip=  1mu plus 1mu minus 1mu \thickmuskip=2mu plus 3mu minus 1mu}
\textmuskip
\def\beq{\dispmuskip\begin{equation}}    \def\eeq{\end{equation}\textmuskip}
\def\beqn{\dispmuskip\begin{displaymath}}\def\eeqn{\end{displaymath}\textmuskip}
\def\bqa{\dispmuskip\begin{eqnarray}}    \def\eqa{\end{eqnarray}\textmuskip}
\def\bqan{\dispmuskip\begin{eqnarray*}}  \def\eqan{\end{eqnarray*}\textmuskip}

\newtheorem{theorem}{Theorem}
\newtheorem{lemma}[theorem]{Lemma}
\newtheorem{proposition}[theorem]{Proposition}
\newenvironment{proof}{{\noindent\bf Proof.}}{\vskip 1ex}
\newenvironment{keywords}{\centerline{\bf\small
Keywords}\begin{quote}\small}{\par\end{quote}\vskip 1ex}

\newdimen\paravsp  \paravsp=1.3ex 
\def\paradot#1{\vspace{\paravsp plus 0.5\paravsp minus 0.5\paravsp}\noindent{\bf\boldmath{#1.}}} 
\def\qmbox#1{{\quad\mbox{#1}\quad}} 

\def\req#1{(\ref{#1})}          
\def\epstr{\epsilon}            
\def\nq{\hspace{-1em}}          
\def\qed{\hspace*{\fill}\rule{1.4ex}{1.4ex}$\quad$\\} 
\def\fr#1#2{{\textstyle{#1\over#2}}} 
\def\frs#1#2{{^{#1}\!/\!_{#2}}} 
\def\e{{\rm e}}                 
\def\v{\boldsymbol}             
\def\G{\Gamma}                  
\def\a{\alpha}
\def\g{\gamma}
\def\t{\theta}

\def\A{{\cal A}}                
\def\X{{\cal X}}                
\def\newx{{\cal N}\!\text{\it ew}}
\def\oldx{{\cal O}\text{\it ld}}
\def\b{\beta}
\def\CL{\ensuremath{\text{CL}}}

\def\M{{\cal M}}

\def\KT{\ensuremath{\text{KT}}}
\def\Bayes{\ensuremath{\text{Bayes}}}
\def\DirM{\ensuremath{\text{DirM}}}
\def\d{{\rm d}} 
\def\nb{\bar\nu}
\def\Rub{\smash{\overline{R}}}
\def\Rlb{\smash{\underline{R}}}
\def\vbs{{\smash{\vec\b^*\!}}} 

\begin{document}

\title{\vspace{-4ex}
\vskip 2mm\bf\Large\hrule height5pt \vskip 4mm
Sparse Adaptive Dirichlet-Multinomial-like Processes
\vskip 4mm \hrule height2pt}
\author{{\bf Marcus Hutter}\\[3mm]
\normalsize Research School of Computer Science \\[-0.5ex]
\normalsize Australian National University \\[-0.5ex]
\normalsize Canberra, ACT, 0200, Australia \\
\normalsize \texttt{http://www.hutter1.net/}
}
\date{May 2013}
\maketitle

\begin{abstract}
Online estimation and modelling of i.i.d.\ data for short
sequences over large or complex ``alphabets'' is a ubiquitous
(sub)problem in machine learning, information theory, data
compression, statistical language processing, and document
analysis. The Dirichlet-Multinomial distribution (also called
Polya urn scheme) and extensions thereof are widely applied for
online i.i.d.\ estimation. Good a-priori choices for the
parameters in this regime are difficult to obtain though. I
derive an optimal adaptive choice for the main parameter via
tight, data-dependent redundancy bounds for a related model. The
1-line recommendation is to set the `total mass' = `precision' =
`concentration' parameter to $m/[2\,\ln{n+1\over m}]$, where $n$
is the (past) sample size and $m$ the number of different symbols
observed (so far).
The resulting estimator
(i) is simple,
(ii) online,
(iii) fast,
(iv) performs well for all $m$, small, middle and large,
(v) is independent of the base alphabet size,
(vi) non-occurring symbols induce no redundancy,
(vii) the constant sequence has constant redundancy,
(viii) symbols that appear only finitely often have bounded/constant contribution to the redundancy,
(ix) is competitive with (slow) Bayesian mixing over all sub-alphabets.
\vspace{4ex} 
\def\contentsname{\centering\normalsize Contents}\setcounter{tocdepth}{1}
{\parskip=-2.7ex\tableofcontents}
\end{abstract}

\vspace*{-2ex}

\section{Introduction}\label{sec:Intro}

The problem of estimating or modelling the probability distribution
of data sequences sampled from an unknown source is central in
machine learning \cite{Bishop:06}, information theory
\cite{Cover:06}, and data compression \cite{Mahoney:12}. I consider
the case where the data items are complex and/or are drawn from a
large space. Many approaches to language modelling and document
analysis \cite{Manning:99} fall into this regime, where data items
are words. Typical documents comprise a small fraction of the
available 100'000+ English words, and words have different
length/complexity/frequency.

\paradot{Online estimation of i.i.d.\ data}
More formally, I consider i.i.d.\ data with base alphabet $\X$ much
larger than the sequence length, which implies that only a small
fraction of symbols (which in case of text are words) appear in the
sequence. I focus on online algorithms that at any time can predict
the probability of the next symbol given only the past sequence and
without knowing the actually used alphabet $\A$ and/or symbol
occurrence frequencies in advance.

While real-word data like text are often not i.i.d, i.i.d.\
estimators are often a key component of more sophisticated models.
For instance, in $n$-gram models, the subsequence of words that
have the same length-$n$ context is (assumed) i.i.d. Since these
subsequences can be very short, good i.i.d.\ estimators for short
sequences and huge alphabet are even more important. The same holds
for variable-order models like large-alphabet context tree
weighting \cite{Tjalkens:93}, and in addition, the employed i.i.d.\
estimators need to be online.

\paradot{Performance measures}
Performance can be measured in many different ways: %
code length \cite{Cover:06}, %
perplexity \cite{Manning:99}, %
redundancy \cite{Wallace:05}, %
regret \cite{Gruenwald:07book}, %
and others. %
The most wide-spread (across disciplines) performance measures are
transformations of the (estimated) data likelihood(s). If
$Q(x_{1:n})$ is the estimated probability of sequence
$x_{1:n}\equiv x_1...x_n$, then $\log 1/Q(x_{1:n})$ is the optimal
code length and $Q(x_{1:n})^{1/n}$ the perplexity of $x_{1:n}$. If
$P$ is some reference measure, then $\log 1/Q - \log 1/P$ is the
redundancy of $Q$ relative to $P$. For log-loss, this is also its
regret, though many variations are used. Many other performance
measures can be upper bounded by (expected) code length
\cite{Hutter:03optisp}. I therefore concentrate on
$-$log-likelihood = code length and redundancy.

\paradot{Dirichlet-multinomial and parameter choice}
The Dirichlet-multinomial distribution is defined as
$\DirM(x_{n+1}=i|x_{1:n})={n_i+\a_i\over n+\a_+}$, which can be
motivated in many ways, e.g.\ by the Polya urn scheme or as below.
This process and extensions thereof like the Pitman-Yor process are
widely studied and applied \cite{Hutter:10pdpx}, in particular for
language processing and document analysis.
Theoretically motivated choices for the Dirichlet parameters $\a_i$ are %
$\a_i=\frs12$ for the Krichevsky-Trofimov (KT) estimator \cite{Krichevsky:81}
and Jeffreys/Bernardo/MDL/MML prior \cite{Jeffreys:46,Jeffreys:61,Bernardo:79,Gruenwald:07book,Wallace:05}, %
$\a_i=0$ for Frequentist and Haldane's prior \cite{Haldane:48}, %
$\a_i=1$ for the uniform/indifference/Bayes/Laplace prior \cite{Bayes:1763,Laplace:1812}, %
and $\a_i=1/|\X|$ for Perks' prior \cite{Perks:47}.%
They are all problematic for large base alphabet $\X$, so is sometimes
optimized or sampled experimentally or averaged with a hyper-prior.
The following table summarizes these choices:
\beq\label{tab:aichoices}
\begin{array}{c||c|c|c|c||c}
  \text{Dirichlet}    & \text{Laplace} & \text{KT\&others}  & \text{Perks} & \text{Haldane} & \text{Hutter} \\ \hline
  \a_i=\displaystyle{\a_+\over|\X|}\rule{0ex}{4ex} & 1 & \displaystyle{1\over 2} & \displaystyle{1\over|\X|} & 0 & \displaystyle {m\over 2|\X|\ln{n+1\over m}}
\end{array}
\eeq
The last column is a glimpse of the results in this paper, where
$m$ is the number of different symbols that appear in $x_{1:n}$.
For continuous spaces $\X$, the Dirichlet process is usually
parameterized by a base distribution $H()$ and a critical
concentration parameter $\b\widehat=\a_+$.

\paradot{Main contribution}
In this paper I introduce an estimator $S$ [Eq.\req{eq:Sdef}],
which essentially estimates the probability of the next symbol by
its past relative frequency, but reserves a small (or large!)
``escape'' probability to new symbols that have not appeared so
far. Such escape mechanisms are well-known and used in data
compression such as prediction by partial match (PPM)
\cite{Clearly:84,Mahoney:12}. This is (somewhat) different from how
the Dirichlet-multinomial regularizes zero frequency with $\a_i>0$
or $\b>0$.

The main contribution is to derive an ``optimal'' escape parameter
$\b^*$ [Eq.\req{eq:bstar}~offline and Eq.\req{eq:bts}~online]. The
key to improve upon existing estimators like the minimax optimal KT
estimator is to consider data-dependent redundancy bounds, rather
than expected or worst-case redundancy, and find its minimizing
$\b$. While the KT estimator and many of its companions have
$\fr12\log n$ redundancy per symbol in $\X$, whether the symbol
occurs in the sequence or not, our new estimator $S^{\b^*}$ suffers
zero redundancy for non-occurring symbols, and essentially only
$\fr12\log n_i+O(1)$ for symbols $i$ appearing $n_i$ times. This is
never much worse and often significantly better than KT.
This also leads to an ``optimal'' variable Dirichlet
parameter $\vbs$. While knowing $\vbs$ is practically useful, the
derived redundancy bounds themselves are of theoretical interest.

\paradot{Contents}
After establishing notation in Section~\ref{sec:Prelim}, I motivate
and state my primary model $S^\b$ in Section~\ref{sec:Model}.
I derive exact expressions and upper and lower bounds for the
redundancy of $S^\b$ for general constant $\b$ in
Section~\ref{sec:Sb}, and show how they improve upon the minimax
redundancy.
I approximately minimize the redundancy w.r.t.\ $\b$ in
Section~\ref{sec:Sbstar}. There are various regimes for the optimal
$\b^*$ and the used alphabet size $|\A|$, even with negative
redundancy.
To convert this into an online model, I make $\b^*$ time-dependent
in Section~\ref{sec:Sbt}, causing very little extra redundancy.
In Section~\ref{sec:Compare} I theoretically compare my models to
the Dirichlet-multinomial distribution, and Bayesian sub-alphabet
weighting.
Section~\ref{sec:Disc} concludes.

Proofs of the lower and two upper bounds can be found in
Appendices~\ref{app:redSbLBnob}, \ref{app:redSbstar}, and \ref{app:redSbt}, %
a derivation of $\b^*$ in Appendix~\ref{app:bopt} %
with improvements in Appendix~\ref{app:bimpr}, %
details of Bayesian subset-alphabet weighting in Appendix~\ref{app:SAW}, %
algorithmic considerations in Appendix~\ref{app:Comp}, %
and an experimental evaluation in Appendix~\ref{app:Exp}. %
Used properties of the (di)Gamma functions can be found in Appendix~\ref{app:Gamma}, %
and a list of used notation in Appendix~\ref{app:Notation}.

\section{Preliminaries}\label{sec:Prelim}

All global notation is introduced in this section and summarized in
Appendix~\ref{app:Notation}.

\paradot{Base alphabet ($\X$, $D$)}
Let $\X$ be the {\em base} alphabet of size $D=|\X|$ from which a
sequence of symbols is drawn. If not otherwise mentioned, I assume
$\X$ to be finite. I have a large base alphabet in mind, but this
is not a technical requirement. The alphabet could literally
consist of e.g.\ ASCII symbols, could be the set of (over 100'000)
English words, or just bits $\{0,1\}$.
Indeed, even finiteness of $\X$ is nowhere crucially used and all
results generalize easily to countable and even continuous $\X$ as
we will see.

\paradot{Total sequence ($n$, $x_{1:n}$, $n_i$)}
I consider sequences $x_{1:n}\equiv (x_1,...,x_n)\in\X^n$ of length
$n$ drawn from $\X$. Let $n_i$ be the number of times, $i$ appears
in $x_{1:n}$. I have in mind that the sequences are sampled
independent and identically distributed (i.i.d), but I actually
never use this assumption. All results in this paper hold for any
individual fixed sequence $x_{1:n}$, and only depend on the order
statistics $\v n=(n_i)_{i\in\X}$.
The crucial parameters are $n$, $D$, the number $m$ of non-zero
counts, and model parameter $\b$ introduced later, which induces
several different regimes, second by the counts $n_i$.

\paradot{Used alphabet ($m$, $\A$, $i,j,k$, $\nu$, $\nb$)}
Only a subset of symbols $\A:=\{x_1,...,x_n\}\subseteq\X$ may
actually appear in a sequence $x_{1:n}$. Our model is primarily
motivated for the regime where the number $m=|\A|$ of {\em used}
symbols is much smaller than $D=|\X|$, as e.g.\ any English text
uses only a small fraction of all possible words.
It turns out that our model can be tuned to actually perform very
well for all possible $1\leq m\leq\min\{n,D\}$: %
constant sequences ($m=1$), %
every symbol appearing only once ($m=n$), %
and all available symbols appear ($m=D$).
Indices $i,j,k$ are understood to range respectively over symbols
in $\X$, $\A$, and $\X\setminus\A$. Without loss of generality I
can assume $i\in\X=\{1,...,D\}$, $j\in\A=\{1,...,m\}$, and
$k\in\X\setminus\A=\{m+1,...,n\}$. I also use $\nb:=n/m$ for the
average multiplicity of symbols in $x_{1:n}$, and $\nu:=m/n$ is its
inverse.

\paradot{Current sequence and observed alphabet
($t$, $x_{1:t}$, $\A_t$, $m_t$, $x_{t+1}$, $n^t_i$, $\newx$, $\oldx$)}
Let $t$ be the current time ranging from $0$ to $n-1$, with
$x_{1:t}$, $\A_t:=\{x_1,...,x_t\}$ and $m_t=|\A_t|$ being
respectively, the sequence, symbols, and number of different
symbols observed so far, and as usual $x_{1:0}=\epstr$ is the empty
string and $\A_0=\{\}$ the empty set. The next symbol to be
predicted or coded is $x_{t+1}=i$. Either $x_{t+1}$ is a new symbol
or an ``old'' symbol. Let $\newx:=\{t=0...n-1:
x_{t+1}\not\in\A_t\}$ and $\oldx:=\{t=0...n-1: x_{t+1}\in\A_t\}$ be
the sets of times for which the next symbol is new/old. Note that
$|\newx|=|\A|=m$. Finally, let $n^t_i$ be the number of times, $i$
appears in $x_{1:t}$.
Note that most inroduced quantities $*_t$ depend on $x_{1:t}$,
but since I consider an (arbitray but) fixed sequence $x_{1:n}$
it is safe to suppress this dependence in the notation.

\paradot{Probability and exchangeability and logarithms
($P$, $Q$, $P^{\rm param}_{\rm name}$, $\ln$)}
$P$ and $Q$ will denote generic probability distributions over sequences,
and $P^{\rm param}_{\rm name}$ specific parameterized and named ones.
For instance, $P^{\v\t}_{iid}$ denotes the model in which symbols are i.i.d.\ with $P(x_t=i)=\t_i$.
Our primary prediction/compression models defined below are $S$, $S^{\b^*}\!\!$, and $S^\vbs$.
%
A distribution $P(x_{1:n})$ is called exchangeable if it is
independent of the order of the symbols in a sequence $x_{1:n}$.
Many distributions have this desirable property \cite{DeFinetti:74}.
%
Since the natural logarithm is mathematically more convenient,
I express all results in `nits' rather than bits.
Conversion to bits is trivial by dividing results by $\ln 2$.

\section{The Main Model}\label{sec:Model}

I am now ready to motivate and formally state our primary model.

\paradot{Derivation of my main model}
My main model is defined via predictive distributions
$S(x_{t+1}|x_{1:t})$ for $t=0...n-1$. If $i$ has appeared $n^t_i$
times in $x_{1:t}$, it is natural to use the past relative
frequency $n^t_i/t$ as the predictive probability that the next
symbol $x_{t+1}$ is $i$. The problems with this are well-known and
obvious: It assigns probability zero and hence infinite log-loss or
code length to any symbol that has not yet been observed.
This problem can be solved by reserving some small (or not so
small) ``escape'' probability $\a_t$ that the next symbol $x_{t+1}$
is new, taken from $n^t_i/t$ by lowering it to $(1-\a_t)n^t_i/t$.
I have to somehow distribute the probability $\a_t$ among the new
symbols $x_{t+1}\in\X\setminus\A_t$. The simplest choice would be
uniform. More generally assign probability $\a_t w^t_k$ to
$k=x_{t+1}\in\X\setminus\A_t$ with $\sum_{k\in\X\setminus\A_t}
w^t_k\leq 1$ and $w^t_k>0$.

One can show that the ansatz above for time-independent weights
leads to an exchangeable distribution if and only if
$\a_t=\b/(t+\b)$ for some constant $\b\geq 0$.

\paradot{Main model}
This motivates our main model
\beq\label{eq:Sdef}
  S(x_{t+1}=i|x_{1:t}) ~:=~ \left\{
  \begin{array}{ccc}
    \displaystyle {n^t_i\over t+\b} & \qmbox{for} & n^t_i>0 \\[2ex]
    \displaystyle {\b w^t_i\over t+\b} & \qmbox{for} & n^t_i=0
  \end{array} \right.
\eeq
for $t=0...n-1$. Note that $S(x_1=i)=w^0_i$ is independent of $\b>0$.
The case conditions can also be written as
$[n^t_{x_{t+1}}>0]\,\equiv\,[x_{t+1}\in\A_t]\,\equiv\,[t\in\oldx]$
and
$[n^t_{x_{t+1}}=0]\,\equiv\,[x_{t+1}\not\in\A_t]\,\equiv\,[t\in\newx]$.
Other motivations and relations to other estimators are given in
Section~\ref{sec:Compare}.

\paradot{Sub-probability}
In general, $\sum_{i\in\X}S(x_{t+1}=i|x_{1:t})\leq 1$, but not
necessarily $=1$. Such sub-probabilities are benign extensions for
many purposes including ours. It is always possible to increase
sub-probabilities to proper probabilities. For $S$ we could replace
$w^t_i$ by $w^t_i/\sum_{k\in\X\setminus\A_t} w^t_k$ as long as
$\X\setminus\A_t$ is not empty, and replace $\b$ by $0$ if ever all
base symbols ($m_t=D$) have appeared. Note that unless $m_t=D$, we
have to assume $\b>0$ to avoid the problems of frequentist
estimation.

\paradot{Sequence probability}
The probability our model assigns to sequence $x_{1:n}$ is
\bqa
  S(x_{1:n}) &=& \prod_{t=0}^{n-1} S(x_{t+1}|x_{1:t})
  ~=~ \prod_{t=0}^{n-1}{1\over t+\b} \prod_{t\in\oldx}n^t_{x_{t+1}} \prod_{t\in\newx}\b w^t_{x_{t+1}} \label{Sjoint0}\\
  &=& \b^{|\A|}  {\G(\b)\over\G(n+\b)} \prod_{t\in\newx}w^t_{x_{t+1}} \prod_{j\in\A}\G(n_j) \label{Sjoint}
\eqa
where $\G$ is the Gamma function.
The symbol count $n^t_j$ increases by 1 for each occurrence of $j$ in the sequence.
Therefore $\prod_{t\in\oldx:x_{t+1}=j}n^t_j=1\cdot...\cdot(n_j-1)=\G(n_j)$,
which establishes the second line.

\section{\boldmath Redundancy of $S^\b$ for General $\b$}\label{sec:Sb}

In this section I motivate and define the concepts of redundancy
and (log-loss) regret and present an exact expression for the
redundancy of $S^\b$ for general constant $\b$. Upper and
lower bounds are easily derived by bounding the involved
Gamma functions. Finally I discuss the $\b$-independent terms in the bound,
and how they improve upon the minimax redundancy.

\paradot{Code length and redundancy/regret}
If a data sequence is sampled from some distribution $P$,
then a lower bound on the expected code length is the entropy $H(P)$ of the source $P$,
which can only be achieved by an encoder which encodes
sequences $x_{1:n}$ in $-\ln P(x_{1:n})$ nits \cite{Shannon:48}.

Arithmetic encoding \cite{Rissanen:76,Witten:87} can (efficiently
and online) achieve this lower bound within 2 bits.
It is therefore appropriate to call
\beqn
  \CL_P(x_{1:n}):=\ln 1/P(x_{1:n})
\eeqn
the (optimal) code length of $x_{1:n}$ (in nits w.r.t.\ $P$).
Arithmetic coding also works for sub-probabilities.

Usually, $P$ is unknown, and one aims at compressors getting close
to $\CL_P$ for all $P$ that might be ``true'' and/or for all $P$
for which it is feasible to do so. Let $\M=\{P\}$ be such a class
of interest; then $\min_{P\in\M}\CL_P(x_{1:n})$ is an (infeasible)
lower bound on the best possible coding if $x_{1:n}$ is sampled
from {\em some} $P\in\M$.

Most modern compressors are themselves based on a (predictive)
distribution $Q$ used together with arithmetic coding
\cite{Mahoney:12}. This motivates the concept of {\em redundancy}
or {\em regret} $R$ as a performance measure for $Q$, which I
define as the difference in code length between the data coded with
predictor $Q$ and the infeasible optimal code length in hindsight:
\beq\label{eq:redDef}
  R_Q(x_{1:n}) ~:=~ \CL_Q(x_{1:n}) - \min_{P\in\M}\CL_P(x_{1:n})
  ~=~ \ln {\max_{P\in\M} P(x_{1:n}) \over Q(x_{1:n})}
\eeq
For comparing the code lengths of different $Q$, any quantity from
which $\CL_Q$ can easily be recovered could be studied: log-loss
regret $\CL_Q-\CL_P$ or redundancy $\CL_Q-H(P)$ where $P$ is the
true distribution of entropy $H(P)$, or $\CL_Q-c$ for any other
``constant'' $c$ independent of $Q$, and of course code length
$\CL_Q$ itself. The redundancy $R_Q$ w.r.t.\ class $\M$ defined
above ($c=\min_{P\in\M}\CL_P(x_{1:n})$) is just often and also here
the most convenient choice.
Upper and lower bounds on redundancies will be denoted by
$\overline R$ and $\underline R$.

\paradot{I.i.d.\ reference class}
As reference class $\M$ I choose the class of i.i.d.\ distributions
with symbol $i\in\X$ having probability $\t_i\in[0;1]$.
\beqn
  P^{\v\t}_{iid}(x_{1:n}) ~:=~ \t_{x_1}\cdot...\cdot\t_{x_n}
  ~=~ \prod_{i\in\X}\t_i^{n_i} ~=~ \prod_{j\in\A}\t_j^{n_j}
\eeqn
The maximum is attained at $\t_i=\hat\t_i:=n_i/n$; therefore
\beq\label{IIDjoint}
  P^{\v{\hat\t}}_{iid}(x_{1:n}) ~=~ \max_{\v\t} P^{\v\t}_{iid}(x_{1:n}) ~=~ n^{-n}\prod_{j\in\A}{n_j^{n_j}}
\eeq

\paradot{Redundancy of $S$}
Subtracting the logarithm of \req{Sjoint} from the logarithm of
\req{IIDjoint} and using abbreviation
$\CL_w(\A):=\sum_{t\in\newx}\ln(1/ w^t_{x_{t+1}})$ discussed below,
one can represent the redundancy of $S$ as follows:

\begin{proposition}[Redundancy of $S$ for constant $\b$]\label{prop:redSb}
For any constant $\b>0$, the redundancy of $S^\b$ relative to the i.i.d.\
class $\M=\{P^{\v\t}_{iid}\}$ can be represented exactly and bounded as follows:
\bqa
  & & \nq\nq \Rlb^\b_S(x_{1:n}) ~\leq~ R^\b_S(x_{1:n}) ~\leq~ \Rub^\b_S(x_{1:n}),\qmbox{where} \nonumber \\
  & & \nq\nq R^\b_S ~=~ \CL_w(\A) - m\ln\b  + \sum_{j\in\A}\ln{n_j^{n_j}\over\G(n_j)} + \ln{\G(n\!+\!\b)\over n^n\,\G(\b)} \label{eq:redSb} \\
  & & \nq\nq \Rub^\b_S ~:=~ \CL_w(\A) - m\ln\b + \sum_{j\in\A}\fr12\ln{n_j\over 2\pi} + n\ln(1\!+\!{\b\over n}) + (\b\!-\!\fr12)\ln({n\over\b}\!+\!1) + 0.082  \label{eq:redSbUB} \\
  & & \nq\nq \Rlb^\b_S ~:=~ \Rub^\b_S - 0.082(m\!+\!2) \label{eq:redSbLB}
\eqa
where $\A\subseteq\X$ are the $m$ (a-priori unknown) symbols
appearing in $x_{1:n}\in\X^n$. The lower bound only holds for $\b\geq 1$.
The 0.082 is actually $1-\ln\sqrt{2\pi}$.
\end{proposition}

The exact expression follows easily by rearranging terms in
\req{Sjoint} and \req{IIDjoint}. The bounds follow from this by
inserting the upper and lower bounds \req{eq:GammaULB} on the Gamma
function and collecting/cancelling matching terms. As can be seen,
the upper and lower bounds only differ by $0.082(m+2)$, hence are
quite tight for small $m$, but loose for large $m$.

In the following paragraphs I discuss the two $\b$-independent
terms. The $\b$-dependent terms will be discussed in the next
section. Note that the following interpretation of \req{eq:redSb}
only refers to code length. The actual way how arithmetic coding
works is very different from this ``naive'' interpretation of the
origin of the different terms in \req{eq:redSb}.

\paradot{Code length of used alphabet $\A$} 
The first term in the redundancy \req{eq:redSb}
\beq\label{eq:CLA}
  \CL_w(\A) ~:=~ \sum_{t\in\newx}\ln(1/ w^t_{x_{t+1}})
\eeq
can be interpreted as follows:
Whenever we see a new symbol $x_{t+1}\not\in\A_t$,
we need to code the symbol itself. This can be done in
$\ln(1/ w^t_{x_{t+1}})$ nits, which together leads to
code length \req{eq:CLA} for the used alphabet $\A$.

A natural choice for the new symbol weights is the uniform
distribution $w^t_i=1/D$ with $\CL_w(\A)=m\ln D$. Since at time $t$
there are only $D-m_t$ new symbols left, we could use normalized
uniform weights $w^t_k=1/(D-m_t)$ with smaller
\beq\label{eq:wUN}
  \CL_w(\A) ~=~ \ln(D)+...+\ln(D-m+1) ~=~ \ln[D!/(D-m)!]
\eeq
For large, structured, and/or infinite alphabet, a more natural
choice is $w^t_i=\exp(-\CL(i))$ with
\beq\label{eq:wCL}
  \CL_w(\A) ~=~ \sum_{t\in\newx}\CL(x_{t+1}) ~=~ \sum_{j\in\A}\CL(j)
\eeq
were new symbols $j$ are {\em somehow} coded (prefix-free) in
$\CL(j)$ nits. For intstance if $\X$ consists of English words,
each word $i$ with $\ell$ letters could be represented as a
byte-string of length $\ell$ plus a 0 terminating byte, hence
$\CL(i)=8\ell+8$.
%
Choice \req{eq:wCL} is interesting since it makes the redundancy
completely independent of the size of the base alphabet, and hence
leads to finite redundancy even for infinite alphabet $\X$.

For all examples of weights above, $\CL_w(\A)$ is independent of
order and timing of new symbols, which justifies suppressing the
dependence on $\newx$. This holds more generally for all $w^t_i$ of
the form $w^t_i=u(i)v(m_t)$
\beq\label{eq:wFact}
  \CL_w(\A) ~=~ \sum_{j\in\A}\ln{1\over u(j)} + \sum_{m'=0}^{m-1}\ln{1\over v(m')}
\eeq
For ease of discussion, I will only consider weights of this form,
and indeed mostly the normalized uniform \req{eq:wUN} and
code-length based \req{eq:wCL} ones. Then also $R^\b_S$ only
depends on the counts $n_i$ but not on the symbol order, as
intended.

\paradot{Code length of relative frequencies $n_i/n$}
Oracle $P^{\smash{\v{\hat\t}}}_{iid}$ predicts symbol $j$ with
empirical frequency $n_j/n$, so $j$ can be coded in $\ln (n/n_j)$
nits. I label an estimator {\sc Oracle} if it relies on extra
information, here, knowing the empirical symbol frequencies in
advance. Technically,
$P^{\smash{\v{\hat\t}(x_{1:n})}}_{\smash{iid}}(x_{1:n})$ is an
inadmissible super-probability. To get a feasible (but offline)
predictor one needs to encode the counts $n_i$ in advance.
Arithmetic coding w.r.t.\ $S^\b$ does not work like that but
imagine it did. The $\ln (n/n_j)$ terms would cancel in the
redundancy leaving a code length for all $n_i$. $\CL(\A)$ tells us
which $n_i$ are zero, so only $n_j$ for $j\in\A$ need to be coded,
which can be done in $\ln n$ nits per $j\in\A$, and the upper bound
\req{eq:redSbUB} suggests possibly even in $\fr12\ln(n_j/2\pi)$
nits.

\paradot{Improvement over minimax redundancy}
It is well known that the minimax redundancy of i.i.d.\ sources is
$\fr12\ln n+O(1)$ per base alphabet symbol
\cite{Rissanen:84,Wallace:05}. My model improves upon this in two
significant ways. Consider the asymptotics $n\to\infty$ in
\req{eq:redSbUB}. First, all symbols $k$ that do not appear in
$x_{1:n}$ induce zero redundancy. Second, each symbol $j$ that
appears only finitely often, induces finite bounded redundancy
$\CL(j)+\fr12\ln{n_j\over 2\pi}$ plus $\b$-terms discussed later.
Only symbols appearing with non-vanishing frequency $n_i/n\not\to
0$ have asymptotic redundancy $\fr12\ln n+O(1)$. This improvement
(a) is possible (only) for specific choices of $\b$ such that the
$\b$-terms are small and (b) was possible by refraining from
deriving a uniform minimax redundancy over all sequences, but one
which depends on the symbol counts.

\paradot{$\b$-independent lower redundancy bound}
In Appendix~\ref{app:redSbLBnob} I derive a $\b$-independent lower
bound on the redundancy that cannot be beaten, whatever $\b$ is
chosen. The following lower bound has the same structure as the
upper bounds I derive later, so the terms will be discussed there.

\begin{theorem}[$\b$-independent lower redundancy bound]\label{prop:redSbLBnob}
For any constant $\b>0$, the redundancy of $S^\b$
is lower bounded uniformly in $\b$ by:
\bqa\label{eq:redSbLBnob}
  & & \nq\nq R^\b_S(x_{1:n}) ~\geq~ \CL_w(\A) - m\ln m + \sum_{j\in\A}\fr12\ln{n_j\over 2\pi} - \fr12\ln n - 0.45m - 0.43
\eqa
\end{theorem}

\section{Redundancy for Approximate Optimal \boldmath$\b^*$}\label{sec:Sbstar}

I am now in a position to approximately minimize the redundancy of
$S^\b$ w.r.t.\ $\b$. Even when only considering asymptotics
$n\to\infty$, I need to distinguish six different regimes for
$\b^*$ depending on how $m$ scales with $n$. I discuss the more
interesting regimes, in particular the unusual situation of
negative redundancy.

\paradot{Optimal constant $\b$}
I now optimize $S^\b$ w.r.t.\ to $\b$. The redundancy $R^\b_S$
is minimized for
\beq\label{eq:dSdb}
  0 ~\stackrel!=~ {\partial R^\b_S\over\partial\b}
  ~=~ -{m\over\b} + \Psi(n\!+\!\b) - \Psi(\b)
\eeq
where $\Psi(x):={\rm d}\ln\G(x)/{\rm d}x$ is the diGamma function.
Neither this equation nor $\partial\Rub^\b_S/\partial\b=0$ have
closed-form solutions, and even asymptotic approximations are a
nuisance. It seems natural to derive expressions for $n\to\infty$
and/or $m\to\infty$, but since $\b$ is inside the diGamma functions
it turns out that considering $\b$-limits leads to fewer cases.
Still one has to separate the regimes $\b\to\infty$, $\b\to
c\lessgtr^{\,\infty}_{\,~0}$, $\b\to 0$, $\b/n\to\infty$, $\b/n\to
c\lessgtr^{\,\infty}_{\,~0}$, and $\b/n\to 0$. I do this in
Appendix~\ref{app:bopt} with further discussion and improvements in
Appendix~\ref{app:bimpr} and stitch together the results, leading
to a surprisingly neat result:

\begin{theorem}[Optimal constant $\b$]\label{thm:optb}
The $\b$ which minimizes $R^\b_S$ \req{eq:redSb} and solves \req{eq:dSdb} is
\beqn
  \b^{min} ~=~ {m\over c_n({m\over n})\ln{n\over m}},\quad
  {\text{where $c_\infty(\nu):=\lim_{n\to\infty}c_n(\nu)$ is smooth and~~~~~~~~}\atop
   \text{monotone increasing from $c_\infty(0)=1$ to $c_\infty(1)=2$.}}
\eeqn
\end{theorem}
%
For $n\gg m$ we have $c_n(m/n)\approx 1$, which suggests the approximation
\beq\label{eq:bstar}
  \b^* ~:=~ {m\over\ln{n\over m}}
\eeq
This has the same asymptotics as $\b^{min}$ in all regimes of
interest and turns out to lead to excellent experimental results.
In practice, $c_n(m/n)$ is closer to 2, so halving $\b^*$ leads to
slightly better results unless $m$ is extremely small. This is due
to a quite peculiar shape of $c_\infty(\nu)$, plotted and discussed
in more detail in Appendix~\ref{app:bimpr}.
The performance difference between $S^{\b^*}$, $S^{\b^*\!\!/2}$,
and $\b^{min}$ are very small though. I hence use $\b^*$
\req{eq:bstar} for most of the theoretical analysis but recommend
$\b^*\!\!/2$ \req{tab:aichoices} in practice.
Since no formal result in this paper explicitly uses that $\b^*$ is
an approximate solution of \req{eq:dSdb}, we can simply take $\b^*$
on faith value and explore its implications.

\paradot{Discussion of $\b^*$} 
The value of $\b^*$ can be intuitively understood in this way: if
$m$ is much larger than $\ln n$, then we will often be coding new
symbols, and therefore we should reserve more probability mass for
them by making $\b$ large. If however $m$ is much smaller than $\ln
n$, coding a new symbol is a rare occurrence, so we use a small
$\b$ to increase the efficiency of coding already previously seen
symbols.
%
More quantitatively, $\b^*$ (and $\b^{min}$) scale with
$n\to\infty$ for various $m$ as follows (where $0<c<\infty$ and
$0\leq\a<1$)
\beq\label{tab:mblim}
\begin{array}{c|c|c|c|c|c|c
}
  m    &    \to c       & \propto\ln n & \propto n^\a & \propto n & \geq n-c    & =n \\ \hline
  \b^* & \sim c/\ln n & \to       c & \propto n^\a/\ln n   & \propto n & \propto n^2 & \infty \\
\end{array}
\eeq
Besides the mentioned $m\smash{\,^\ll_\gg}\ln n$ divide, note that
if most symbols appear only once, then $\b\propto n^2$ grows very
rapidly. On the other hand $\b^*$ is never very small: $1/\ln n$ is
a lower bound, even if $m=1$. If no symbol appears twice, then
$\b^*=\infty$ is obviously the best choice. Appendix~\ref{app:Exp}
shows that $S^{\b^*}$ works very well in all six regimes.

I also tried ``minor'' modifications but theory breaks down for
some, and experiments for others. The only leeway, apart from
replacing $c_n()$ by a constant in $[1;2]$ I could find is adding
or subtracting small constants from $m$ and/or $n$ in
\req{eq:bstar}. This will later be used to regularize $\b^*$ for
$m=n$.
Note that $\b^*$ depends on the a-priori unknown $n$ and $m$, so
$S^{\b^*}$ is not online. This will be rectified in
Section~\ref{sec:Sbt}. In Appendix~\ref{app:redSbstar} I prove the
following redundancy bound:

\begin{theorem}[Redundancy of $S$ for ``optimal'' constant $\b^*$]\label{prop:redSbstar} 
The redundancy of $S^{\b^*}$ with $\b^*=m/\ln{n\over m}$ is bounded by
\beq\label{eq:redSbstar}
  R^{\b^*}_S(x_{1:n}) ~\leq~
  \CL_w(\A) - (m-\fr12)\ln m + \sum_{j\in\A}\fr12\ln n_j -\fr12\ln n + m\ln\ln{\e n\over m} + 0.56m + 0.082
\eeq
\end{theorem}

\paradot{Discussion of $R^{\b^*}_S$}
The first and third term have already been discussed.
The second term is the most important one for large $m$.
%
It is about $-\ln\G(m)-m+1$ by \req{eq:GammaULB}.
Therefore for uniform normalized weights \req{eq:wUN} we get
\beq\label{eq:CLDm}
  \CL_w(\A) - (m-\fr12)\ln m ~=~ \ln{D\choose m} - m + \ln m + 1 ~\left\{ {-0.082\atop +0~~~~~~} \right.
\eeq
There are ${D\choose m}$ ways of choosing $m$ symbols out of $D$,
therefore $\ln{D\choose m}$ corresponds to the optimal uniform code
length for the used unordered alphabet. At first, $S^{\b^*}$ seemed
to be more wasteful, coding the $m'$th new symbol in $\ln(D-m'+1)$
nits, hence codes $\A$ including order in $\CL_w(\A)$ nits. But
through the back door by a suitable choice of $\b$, it actually
achieves the theoretically optimal uniform code length
$\ln{D\choose m}$ for the used alphabet, plus other smaller terms.
For large $m$, this can be significantly smaller than $\CL_w(\A)$.

In the extreme case of $m=D$, we have $\ln{D\choose D}=0\ll D\ln
D$. If also $n=m$, we have $\CL_w(\A)=\ln n!$ and $n_i=1\forall i$
and hence
\beqn
  R^{\b^*}_S ~\leq~ \ln n! - n\ln n + 0.56n+0.082 ~\leq~ \fr12\ln n - 0.44n + 1.082
\eeqn
which is negative for $n>4$. This is not a contradiction. It just
says that in this case $S$ codes better than oracle
$P^{\smash{\v{\hat\t}}}_{iid}=({1\over n})^n$. Indeed, if we know
that every symbol appears exactly once, we can code their
permutation in $\ln n!$ rather than $n\ln n$ nits. The $+0.56n$
slack is an artefact of our bound, not of $S^{\b^*}\!\!$, and can
be improved to $0.082n$. The argument generalizes to large $m<n$.

In the other extreme of a constant sequence $x_t=j\forall t$, we
have $m=1$, $P^{\smash{\v{\hat\t}}}_{iid}=1$, $\b^*=1/\ln n$ and
$\CL_{S^{\b^*}}=R^{\b^*}_S\to\CL_w(j)+1$ for $n\to\infty$, i.e.\ 1
nit above theoretical optimum from \req{eq:redSb} and
$R^{\b^*}_S\leq\CL_w(j)+\ln\ln(\e n)+0.65$ from \req{eq:redSbstar},
i.e.\ asymptotically there is only $\ln\ln n$ nits slack in the
bound. This argument generalizes to constant $m>1$.

\section{Redundancy for Variable \boldmath $\vbs$}\label{sec:Sbt}

Since the optimal $\b^*=m/\ln{n\over m}$ depends on $m$ and $n$,
$S^{\b^*}$ cannot be used online, which defeats one of its purposes
and significantly limits its application as discussed in the
introduction. I rectify this problem by allowing a time-dependent
$\b$ in my model, and by adapting $\b^*$ in (nearly) the most
obvious way. I derive a redundancy bound for this variable $\vbs$
which for small $m$ is only slightly worse than the previous one
for constant $\b^*$.

\paradot{Choice of $\vbs$}
A natural way to arrive at an online algorithm is to replace $n$ by
$t$ and $m$ by $m_t$, both known at time $t$ and converging to $n$
and $m$ respectively. This leads to a time-dependent `variable'
$\b_t=m_t/\ln{t\over m_t}$. This works fine except if $m_t=t$, in
which case $\b_t=\infty$ assigns zero probability that the next
symbol is an old one. This is unacceptable, since $m_t=t$ is
typical for small $t$.

If we are at time $t$, we use $\b_t$ to predict $x_{t+1}$ so should
assume that the sequence has (at least) length $t+1$, which
suggests $\b_t=m_{t+1}/\ln{t+1\over m_{t+1}}$. The problem here is
that $m_{t+1}$ depends on the unknown $x_{t+1}$, and technically
$S$ becomes an (unusable) super-probability. Since $m_{t+1}=m_t$ if
$x_{t+1}$ is old anyway, a natural choice is
$\b^*_t=m_t/\ln{t+1\over m_t}$, which still has the same
asymptotics \req{tab:mblim} as $\b^*$, except for $m_t=t$ it is
finite and grows with $t^2$. For $t=0$ I define $S(x_1=i)=w^t_i$ or
equivalently choose any $0<\b^*_0<\infty$. For convenience I
summarize the adaptive model with parameters and definitions in the
box on the next page.

\begin{figure}
\noindent\fbox{
\begin{minipage}{0.97\textwidth}
The $S^\vbs$-probability of $x_{t+1}=i\in\X$ given $x_{1:t}\in\X^t$ is defined as
\beq\label{eq:Svbs}
  S^\vbs(x_{t+1}=i|x_{1:t}) ~:=~ \left\{
  \begin{array}{ccc}
    \displaystyle {n^t_i\over t+\b^*_t} & \qmbox{for} & n^t_i>0 \\[2ex]
    \displaystyle {\b^*_t w^t_i\over t+\b^*_t} & \qmbox{for} & n^t_i=0
  \end{array} \right.
\eeq
\beq\label{eq:bts}
  \b^*_t ~:=~ {m_t\over\ln{t+1\over m_t}},\quad t\geq 1, \qquad 0<\b^*_0<\infty \text{ (any)}, \qquad \vec\b:=(\b_0,\b_1,\b_2,...)
\eeq
\beqn
  \sum_{k\in\X\setminus\A_t} w^t_k \leq 1, ~~\qmbox{e.g.}~~ w^t_i ~=~ {1\over D-m_t} ~~\qmbox{or}~~ w^t_i ~=~ \e^{-\CL(i)}
\eeqn
\beqn
  m_t=|\A_t|, \qquad \A_t=\{x_1,...,x_t\}, \qquad n^t_i=|\{\tau\in\{1,...,t\}:x_\tau=i\}| \\[1ex]
\eeqn\end{minipage}
}\vspace{-2ex}
\end{figure}
Note that compact representation \req{Sjoint} does not hold
anymore: The resulting process $S^\vbs(x_{1:n})$ is no longer
exchangeable, but close enough in the sense that a comparable upper
bound as for $\b^*$ holds. The constants are somewhat worse, but
mostly due to the crude proof (see Appendix~\ref{app:redSbt}).

\begin{theorem}[Redundancy of $S$ for ``optimal'' variable $\vbs$]\label{prop:redSbt} 
The redundancy of $S^\vbs$ with $\b^*_t=m_t/\ln{t+1\over m_t}$ is bounded by
\beq\label{eq:redSbt}
  R^\vbs_S(x_{1:n})
   ~\leq~ \CL_w(\A) - (m\!-\!1)\ln m + \sum_{j\in\A}\fr12\ln n_j - \fr12\ln n
   + \fr32 m\ln\ln\fr{2n}m + 2.33m + 0.86
\eeq
\end{theorem}

The bounds \req{eq:redSb}, \req{eq:redSbUB}, \req{eq:redSbstar},
and \req{eq:redSbt}, except for the first term, are independent of
the base alphabet size $D$. For $w^t_i=2^{-\CL(i)}$, the bounds are
completely independent of $D$. They therefore also hold for
countably infinite alphabet. Analogous to the Dirichlet-multinomial
generalizing to the Chinese restaurant process, $S$ can also be
generalized to continuous spaces $\X$. The weights $w^t_i$ become
(sub)probability densities ($\int_{\X\setminus A}w^t_i\d i\leq 1$).
The bounds remain valid, we only lose the code length
interpretation of $\CL_w(\A)$.

\paradot{Proof idea}
Unlike in \req{eq:redSb} for constant $\b$, $R^\vbs_S$ depends on
the order of symbols and cannot be expressed in terms of Gamma
functions bound by \req{eq:GammaULB}.
Furthermore,  $\b^*_t$ is generally not monotone in $t$, nor does
it factor into monotone increasing and/or decreasing functions,
which makes the analysis cumbersome but not impossible due to a
different special property of $\b^*_t$.
I show that by swapping two consecutive symbols, $x_t$ being
$\oldx$ and $x_{t+1}$ being $\newx$, the redundancy always
increases. It is therefore sufficient to upper bound $R^\vbs_S$ for
sequences in which all new symbols come first before they repeat.
For such a sequence, by separating $t\leq m$ for which $m_t=t$ and
$t\geq m$ for which $m_t=m$, it is then possible to upper bound the
handfull of resulting sums.

\section{Comparison to Other Methods}\label{sec:Compare}

In this section I theoretically (and in Section~\ref{app:Exp}
experimentally) compare our models to various other more or less
related ones, namely, the Dirichlet-multinomial with KT and Perks
prior, and Bayesian sub-alphabet weighting. An experimental
comparison can be found in Appendix~\ref{app:Exp}.

\paradot{Dirichlet-multinomial distribution}
The Dirichlet distribution
\beqn
  \text{Dir}^{\v\a}(\v\t) ~:=~ {\G(\a_+)\over\prod_i\G(\a_i)}\prod_{i=1}^D\t_i^{\a_i-1}
\eeqn
with parameters $\a_i>0$ and $\a_+:=\a_1+...+\a_D$
used as a Bayesian prior for $P^{\v\t}_{iid}$ leads to joint and predictive
Dirichlet-multinomial distribution
\bqan
  \DirM^{\v\a}(x_{1:n}) &:=& \int P^{\v\t}_{iid}(x_{1:n})\text{Dir}^{\v\a}(\v\t)\d\v\t
  = {\G(\a_+)\prod_i\G(n_i\!+\a_i)\over\G(n\!+\!\a_+)\prod_i\G(\a_i)}, \nonumber \\
  \DirM^{\v\a}(x_{t+1}=i|x_{1:t}) &=& {n^t_i+\a_i\over t+\a_+} ~~~~~~~~~~\qmbox{with redundancy} \nonumber
\eqan\vspace{-3ex}
\bqa\label{eq:redDiri}
  R^{\v\a}_\DirM(x_{1:n}) &=& \sum_{i=1}^D \ln{n_i^{n_i}\G(\a_i)\over\G(n_i\!+\a_i)} - \ln{n^n\G(\a_+)\over\G(n\!+\!\a_+)}
  \\ \label{eq:redDiriLimit}
  & \stackrel{n_i\to\infty}\longrightarrow &
  {D\!-\!1\over 2}\ln{n\over 2\pi} + \sum_i(\fr12\!-\!\a_i)\ln{n_i\over n} + \sum_i\ln\G(\a_i)-\ln\G(\a_+)
\eqa
If we choose constant weights $w^t_i=\a_i/\a_+$ and $\b=\a_+$  in $S$,
we see that $\DirM(x_{t+1}=i|x_{1:t})$ is the sum of both cases in \req{eq:Sdef},
hence $\DirM(x_{t+1}=i|x_{1:t})\geq S(x_{t+1}=i|x_{1:t})$.
Therefore, the upper redundancy bound in Proposition \ref{prop:redSb} also holds for
\DirM: $R^{\v\a}_\DirM \leq R^{\a_+}_S[w^t_i:=\a_i/\a_+]\leq$ Eq.\req{eq:redSbUB}.
The analysis in Section \ref{sec:Sbstar} suggests to set
the Dirichlet parameters to $\a^*_i:=w^0_i\b^*$ for which
$R^{\v\a^*}_\DirM \leq R^{\b^*}_S[w^t_i:=\a_i/\a_+] \leq$ Eq.\req{eq:redSbstar}.
If we allow for time-dependent $\a_i$, Section \ref{sec:Sbt} suggests
to set $\a_i=\a^{t*}_i:=w^t_i\b^*_t$
for which $R^{\vec{\v\a}^*}_\DirM \leq R^\vbs_S\leq$ Eq.\req{eq:redSbt},
but note that weights $w^t_i$ must normalize
over $\X$ rather than $\A_t$ for \DirM\ to form a (sub)probability.
This can harm performance but only for large $m$.
Note that for continuous $\X$ and weight {\em density} $w()$,
$S$ and $\DirM$ coincide.

The overall suggestion {\em if} using the (adaptive)
Dirichlet-multinomial for prediction or compression or estimation
is to choose variable parameters
\beq\label{eq:astar}
  \a_i ~=~ \a^{t*}_i ~:=~ {m_t\over D\ln{t+1\over m_t}} \qmbox{or}
  {2^{-\CL(i)}m_t\over \ln{t+1\over m_t}}
\eeq

\paradot{The KT estimator}
As can be seen from \req{eq:redDiriLimit}, for $\a_i=\fr12$ the
\DirM\ redundancy \req{eq:redDiri} is asymptotically independent of
the counts $(n_i)$, and indeed it is well-known that asymptotically
this is essentially also the best choice for the worst counts
\cite{Krichevsky:81,Krichevskiy:98,Wallace:05}. This so-called
KT-estimator has minimax redundancy \cite{Begleiter:06}
\beq\label{eq:KTUB}
  R^{\v 1/2}_\DirM ~\leq~ {D-1\over 2}\ln n + \ln D
\eeq
Asymptotically, this bound is essentially tight. %
We can compare this to our bound \req{eq:redSbstar}.
For $m\ll n$, the dominant term in \req{eq:redSbstar} is $\sum_j\fr12\ln n_j$.
This can be bounded by Jensen's inequality as
\beqn
  \sum_{j\in\A}\fr12\ln n_j - \fr12\ln n
  ~\leq~ {m-1\over 2}\ln {n\over m}
  ~\leq~ {m-1\over 2}\ln n
  ~\leq~ {D-1\over 2}\ln n
\eeqn
so is clearly much smaller than \req{eq:KTUB} due to symbols that
do not appear (gap in the third inequality) and symbols that appear
rarely (gap in the first+second inequality). The latter happens
often in particular for large $m$, but then the other terms in
\req{eq:redSbstar} gain relevance.

\paradot{Sparse KT estimators}
If we knew the used alphabet $\A$ in advance, we could employ the
\KT\ estimator on this sub-alphabet without reference to the base
alphabet $\X$ and achieve much smaller redundancy $\leq {m-1\over
2}\ln n+\ln m$.
%
In absence of such an oracle, we could code unordered $\A$ in
advance in $\ln{D\choose m}$ nits, which gives an off-line
estimator with $\leq D\ln{D\over m}$ extra redundancy above the
oracle.
%
We can even get online versions: A light-weight way is at time $t$
to use \KT\ on $\A_t$ but reserve an escape probability of
$\fr1{t+1}$ for and uniformly distribute it among the unseen
symbols $\X\setminus\A_t$, which leads to a similar but larger
extra redundancy of $\ln n+m\ln D+m+\ln 2$ \cite{Hutter:12ssdc}.
%
A heavy-weight Bayesian solution is to take a weighted average over
the $\KT_{\A'}$ estimators for all $\A'\subseteq\X$
\cite{Tjalkens:93}. As prior one could take a uniform distribution
over the size $m'$ of $\A'$, and then for each $m'$ a uniform
distribution over all $\A'$ of size $m'$ with extra redundancy
$\leq D\ln{D\over m}+\ln D$. The resulting exponential mixture can
be computed in linear time in $D$ as discussed in
Appendix~\ref{app:SAW}. This is still a factor of $D$ slower than
all other estimators considered in this paper. Otherwise the
linear-time update rule has a similar structure to \req{eq:Svbs},
and hence $S^\vbs$ may be derivable as an approximation to Bayesian
sub-alphabet weighting.

\section{Conclusion}\label{sec:Disc}

I introduced and analyzed a model, closely related to the
Dirichlet-multinomial distribution, which predicts an $\oldx$
symbol with its past frequency scaled down by $t\over t+\b$ and a
new symbol with its weight, scaled down by $\b\over t+\b$. Natural
weight choices are uniform and $2^{-\text{CodeLength}}$.

I derived exact expressions and for small $m$ rather tight bounds
for the code length and redundancy. The bounds were data-dependent
rather then expected or worst-case bounds. This led to an
(approximately) optimal choice of $\b$ different from traditional
recommendations. The constant offline $\b^*$ \req{eq:bstar} depends
on the total sequence length $n$ and number of different used
symbols $m$. The variable online $\vbs$ \req{eq:bts} depends on the
current sequence length $t$ and number of different symbols
observed so far $m_t$.

The redundancy bounds additionally depend on the individual symbol
counts $n_i$ themselves. They show that $S^{\b^*}$ has (at most)
zero redundancy for unused symbols and finite redundancy for
symbols occurring only finitely often, unlike the KT estimator and
companions which have redundancy $\fr12\ln n+O(1)$ per base symbol,
whether it occurs or not. Indeed, my bounds are independent of the
base alphabet size $D$, therefore also hold for denumerable and
with suitable reinterpretation for continuous $\X$.

There seems to be not much leeway in choosing a globally good $\b$.
Experimentally it seems that even slight changes in $\b^*$ can
significantly deteriorate performance in some $(m,n,D)$-regime, but
can only marginally and locally improve performance in others.
Empirically $S^\vbs$ seems superior to the other fast online
estimators I compared it to. See Appendix~\ref{app:Exp} for some
results.

As a simple, online, fast, i.i.d.\ estimator, $S^\vbs$ should be a
useful alternative sub-component in more sophisticated (online)
estimators/predictors/compressors/modellers such as large-alphabet
CTW \cite{Tjalkens:93} and others
\cite{Hutter:12ctswitch,Hutter:12adapctw,Mahoney:12}. The derived
redundancy bounds are of theoretical interest, not only for
optimizing model parameters.

\paradot{Acknowledgements}
I thank the anonymous reviewers for valuable feedback, and in
particular one reviewer for providing the efficient representation
of the Bayesian sub-alphabet estimator in Appendix~\ref{app:SAW}.


\section*{References}\label{sec:Bib}
\addcontentsline{toc}{section}{\refname}
\def\refname{\vspace{-4ex}}

\begin{small}

\end{small}

\appendix 
\section{Approximations of the (Di)Gamma Function}\label{app:Gamma}

\beq\label{eq:GammaULB}
  (x-\fr12)\ln x-x+\ln\sqrt{2\pi}
  \mathop{\leq}\limits_{\scriptstyle\uparrow\atop\scriptstyle\forall x>0}
  \ln\G(x)
  \mathop{\leq}\limits_{\scriptstyle\uparrow\atop\scriptstyle\forall x\geq 1}
  (x-\fr12)\ln x-x+1
\eeq
The lower bound is asymptotically sharp for $x\to\infty$ but a
factor of 2 too small for $x\to 0$. The absolute error of upper and
lower bound for all $x\geq 1$ is at most
$1-\ln\sqrt{2\pi}\,\dot=\,0.081$. Some other used identities,
asymptotics, and bounds are:
\beq\label{eq:lntsum}
  \sum_{t=1}^{n-1}\ln t ~=~ \ln\G(n)
\eeq
\beq\label{eq:lnbnd}
  1-1/x ~\leq~ \ln x ~\leq~ x-1 ~~~~~~ [= ~\text{iff}~ x=1]
\eeq
\beq\label{eq:Psi}
  \Psi(z) ~=~ {\d\ln\G(z)\over\d z} ~\sim~ \ln z - O\Big({1\over z}\Big)
\eeq
\beq\label{eq:Gsxb}
  \G(z) ~\leq~ {1\over z} \qmbox{for} z ~\leq~ 1
\eeq

\section{Proof of Theorem~\ref{prop:redSbLBnob}}\label{app:redSbLBnob}

I start with the lower bound \req{eq:redSbLB} rewritten as
\beq\label{eq:redSbLBApp}
  \Rlb^\b_S ~=~ \CL_w(\A) - m\ln\b + \sum_{j\in\A}\fr12\ln n_j + n\ln(1\!+\!{\b\over n}) + (\b\!-\!\fr12)\ln({n\over\b}\!+\!1) - m - [1-\ln\sqrt{2\pi}]
\eeq
which is valid for $\b\geq 1$. Let
\beqn
  R(\b) :=- m\ln\b + n\ln(1\!+\!{\b\over n}) + (\b\!-\!\fr12)\ln({n\over\b}\!+\!1)
\eeqn
be the $\b$-dependent terms in \req{eq:redSbLBApp}.

For $\underline{1\leq\b\leq n}$,
\beqn
  R(\b) ~\geq~ -m\ln\b + (\b\!-\!\fr12)\ln 2 ~\geq~ - m\ln m + m\ln\ln 2
\eeqn
The last inequality follows from minimizing the first w.r.t.\ $\b$
by differentiation and inserting the minimizer $\b=m/\ln 2$ and
dropping the second term.

For $\underline{\b\geq n}$ and with abbreviations $z:=n/\b\leq 1$
and $\nb={n\over m}\geq 1$ we get
\bqan
  R(\b) &\geq& -m\ln\b + n\ln{\b\over n} + {\b\over 2}\ln\Big({n\over\b}+1\Big) \\
  &=& (n-m)\ln\b -n\ln n + {n\over 2}{\ln(1+z)\over z}
  ~~~~~~~~~~~~~~~~~~~~ \left[{\text{increasing in $\b$}\atop\text{decreasing in $z$}}\right] \\
  &\geq& (n-m)\ln n - n\ln n +{n\over 2}\ln 2 \\
  &=& -m\ln m + m[\nb\fr12\ln 2-\ln\nb]
  ~~~~~~~~~~~~~~~~~~~~ [\text{minimized for $\nb=2/\ln 2$}] \\
  &\geq& -m\ln m + m[1-\ln 2+\ln\ln 2] \\
  &\geq& -m\ln m + m\ln\ln 2
\eqan
which is the same as for $1\leq\b\leq n$.
Plugging this into \req{eq:redSbLBApp} we get for $\b\geq 1$
\beq\label{eq:Rlbbg1}
  R^\b_S(x_{1:n}) ~\geq~ \CL_w(\A) - m\ln m + \sum_{j\in\A}\fr12\ln n_j
  - m[1-\ln\ln 2] - [1 - \ln\sqrt{2\pi}]
\eeq

For $\underline{\b\leq 1}$ we need to start with the exact
expression \req{eq:redSb}:
\bqan
  \sum_{j\in\A} \ln{n_j^{n_j}\over\G(n_j)}
  &\stackrel{\req{eq:GammaULB}}\geq& \sum_{j\in\A} [\fr12\ln n_j+n_j-1]
  ~=~ \sum_{j\in\A} [\fr12\ln n_j] + n - m
\\
  -m\ln\b + \ln{1\over\G(\b)} &\stackrel{\req{eq:Gsxb}}\geq& (m\!-\!1)\ln{1\over\b} ~\geq~ 0
\\
  \ln{\G(n+\b)\over n^n} &\stackrel{\req{eq:GammaULB}}\geq&
  (n+\b-\fr12)\ln(n+\b)- (n+\b) + \ln\sqrt{2\pi}-n\ln n \\
  &=& n\ln(1+{\b\over n}) + (\b-\fr12)\ln(n+\b)-(n+\b) + \ln\sqrt{2\pi} \\
  &\geq& -\fr12\ln(2n) - n - 1 +\ln\sqrt{2\pi}
\eqan
Putting everything together we get for $\b\leq 1$
\beq\label{eq:Rlbbl1}
  R^\b_S(x_{1:n}) ~\geq~ \CL_w(\A) + \sum_{j\in\A}\fr12\ln n_j -\fr12\ln n - m - [1-\ln\sqrt{2\pi}+\fr12\ln 2]
\eeq
Pairing up terms (sometimes zero) in \req{eq:Rlbbg1} and
\req{eq:Rlbbl1} and always taking the smaller one, we get after
some rewrite \req{eq:redSbLBnob}, valid for all $\b$.

\section{Derivation of Approximate Optimal \boldmath$\b^*$}\label{app:bopt}

\paradot{Exact implicit expression}
The redundancy of $S$ is minimized for
\beq\label{eq:dSdb2}
  0 ~\stackrel!=~ {\partial R^\b_S\over\partial\b}
  ~=~ -{m\over\b} + \Psi(n\!+\!\b) - \Psi(\b)
\eeq
where $\Psi(x):={\rm d}\ln\G(x)/{\rm d}x$ is the diGamma function.
Our goal is to approximately solve this equation w.r.t.\ $\b$.
Since no formal result in this paper explicitly uses that $\b^*$ is
an approximate solution of \req{eq:dSdb2}, I only motivate the form
of $\b^*$ by asymptotic considerations without discussing the
accuracy of the approximation for finite $n$.
With the following change in variables
\beqn
  0 < z:={n\over\b} < \infty ~~\qmbox{and}~~ 0<\nu:={m\over n} < 1
\eeqn
\req{eq:dSdb2} can be written as
\beqn
  \nu = {1\over z}\Big[\Psi\big(n(1\!+\!{1\over z})\big) - \Psi\big({n\over z}\big)\Big]
\eeqn
We need to solve this w.r.t.\ $z$ for large $n$.

\paradot{$\b\to c>0$ and $\b\to\infty$}
\beqn
  \underline{\b\to\infty} ~~~\Longrightarrow~~~ {n\over z}\to\infty
  ~~~\stackrel{2\times\req{eq:Psi}}\Longrightarrow~~~
  \nu \to {1\over z}\Big[\ln\big(n(1\!+\!{1\over z})\big) - \ln\big({n\over z}\big)\Big]
  ~=~ {1\over z}\ln(1+z)
\eeqn
which is actually good for any $n$ as long as $z=o(n)$. Next consider
\beqn
  \underline{\b\to c}>0 ~~~\Longrightarrow~~~ {n\over z}\to c
  ~~~\stackrel{\req{eq:Psi}}\Longrightarrow~~~
  \nu \to {1\over z}\Big[ \underbrace{\ln(1\!+\!z)}_{\sim\ln n\to\infty}
  + \underbrace{\ln\big({n\over z}\big)-\Psi\big({n\over z}\big)}_{\to\ln(c)-\Psi(c)=const.} \Big]
  ~\sim~ {1\over z}\ln(1+z) \vspace{-3ex}
\eeqn
Therefore we need to solve
\beqn
  \nu ~=~ g(z) := {1\over z}\ln(1+z) \qmbox{for} z=O(n), \qquad 0<z<\infty, \qquad 0<\nu<1
\eeqn
i.e.\ invert function $g$.

\begin{lemma}[Inverse of \boldmath$\ln(1+z)/z$]\label{lem:invln1pzdz}
The function $g(z):={1\over z}\ln(1+z)$ with domain $0<z<\infty$ is
strictly monotone decreasing and has inverse
$g^{-1}(\nu)={c(\nu)\over\nu}\ln{1\over\nu}$ with domain $0<\nu<1$,
where $c(\nu)$ is smooth and strictly monotone increasing from
$c(0^+)=1$ to $c(1^-)=2$.
\end{lemma}

\begin{proof}
Strict monotonicity of $g$ and therefore existence of an inverse
follows from
\beqn
  g'(z) ~=~ {1\over z^2}\Big[{z\over 1+z}-\ln(1+z)\Big] ~\stackrel{\req{eq:lnbnd}}<~ 0
\eeqn
I first study the asymptotics of $\nu=g(z)$ for $z\to 0$ and
$z\to\infty$.
\bqan
  \underline{z\to 0} ~~&\Longrightarrow&~~ \nu\to 1, \qmbox{more precisely}
  \nu = 1-\fr12 z+O(z^2) ~~~\Longrightarrow~~~ z\approx 2(1-\nu)
\\
  \underline{z\to\infty} ~~&\Longrightarrow&~~ \nu\to 0,
\qmbox{and asymptotically}
  z \approx {1\over\nu}\ln{1\over\nu}
\eqan
I got the last expression by fixed point iteration: Rewrite $\nu=g(z)$ as $z={1\over\nu}\ln(1+z)$
and now iterate $z_{t+1}={1\over\nu}\ln(1+z_t)$ starting from any $0<z_0:=c<\infty$.
This gives $z_1={1\over\nu}\ln(1+c)$ and
\beqn
  z_2 ~=~ {1\over\nu}\ln\Big[1+\!\!\!\underbrace{1\over\nu}_{\to\infty}\!\!\!\ln(1\!+\!c)\Big]
  ~\sim~ {1\over\nu}\ln\Big[{1\over\nu}\ln(1\!+\!c)\Big]
  ~=~ {1\over\nu}\Big[\underbrace{\ln{1\over\nu}}_{\to\infty}+\underbrace{\ln\ln(1\!+\!c)}_{finite}\Big]
  ~\sim~ {1\over\nu}\ln{1\over\nu}
\eeqn
No more iterations are needed!
If we tentatively apply the $\nu\to 0$ expression for $\nu\to 1$ we get
\beqn
  z ~\sim~ {1\over\nu}\ln{1\over\nu} ~=~ (1\!-\!\nu) + O((1\!-\!\nu)^2) ~\to~ 0 \qmbox{for} \nu\to 1
\eeqn
The limit value is right, but the slope is $\frs12$ of what it should be.
${2\over\nu}\ln{1\over\nu}$ would have the right slope at $\nu=1$.
Therefore
\beqn
  z ~=~ {c(\nu)\over\nu}\ln{1\over\nu} \qmbox{for some function $c(\nu)$ with $c(0^+)=1$ and $c(1^-)=2$}
\eeqn
which suggests that $c(\nu)$ might always lie in interval $[1;2]$.
I prove this by showing that $c(\nu)$ is a monotone increasing function of $\nu$.
\beqn
  \qmbox{From}
  g^{-1}(\nu)={c(\nu)\over\nu}\ln{1\over\nu}
  \qmbox{we get}
  c(\nu)={\nu g^{-1}(\nu)\over \ln(1/\nu)}
\eeqn
Since $g(z)$ is smooth, also $g^{-1}(\nu)$ and $c(\nu)$ are smooth.
Since $g()$ is monotone decreasing, rather than proving $c()$ to be increasing,
it is equivalently to show that
\beqn
  f(z) ~:=~ c(g(z)) ~=~...~=~ {\ln(1\!+\!z)\over\ln z-\ln\ln(1\!+\!z)}
\eeqn
is monotone decreasing in $z$. For this, it is sufficient to show
\beqn
  0 ~>~ f'(z) ~=~ ... ~=~ {\ln z-\ln\ln(1\!+\!z)-{1+z\over z}\ln(1\!+\!z)+1 \over (1\!+\!z)[\ln z-\ln\ln(1\!+\!z)]^2}
  ~=:~ {h(z)\over\text{denominator}}
\eeqn
Since $h(0^+)=0$, it is sufficient to show $h'(z)<0$:
\beqn
  h'(z) ~=~...~=~ {[\ln(1\!+\!z)]^2-{z^2\over 1+z} \over z^2\ln(1\!+\!z)} ~<~ 0
        ~~~~\iff~~~~ r(z) ~:=~ \ln(1\!+\!z)-{z\over\sqrt{1\!+\!z}} ~<~ 0
\eeqn
Since $r(0^+)=0$, it is sufficient to show $r'(z)<0$:
\beqn
  r'(z) ~=~...~=~ {\sqrt{1\!+\!z}-(1+\fr12 z)\over(1\!+\!z)^{3/2}} ~<~ 0,
  \qmbox{which is true, since} 1\!+\!z ~<~ (1\!+\!\fr12 z)^2
\vspace{-2ex}
\eeqn
\qed
\end{proof}

\paradot{Approximation of $c(\nu)$}
In Appendix~\ref{app:bimpr} I discuss approximations for $c(\nu)$.
In the main text I simply replace $c(\nu)$ by 1, i.e.\
$z={1\over\nu}\ln{1\over\nu}$ which has the right asymptotics for
the $\nu\to 0$ ($m\ll n$) regime I am primarily interested in and
still the right limit for $\nu\to 1$. I also found that this
choice is consistent with the other regimes in \req{tab:mblim}, in
particular with $\b\to 0$. Back in $(n,m,\b)$ notation we get
\beqn
  \b ~=~ {n\over z} ~=~ {n\over{1\over\nu}\ln{1\over\nu}} ~=~ {m\over\ln{n\over m}} ~=:~ \b^*
\eeqn

\paradot{$\b\to 0$}
I finally consider the $\b\to 0$ regime.
Using the general recurrence $\Psi(\b)=\Psi(\b+1)-{1\over\b}$ in \req{eq:dSdb2} we get
\beqn
  0 ~=~ -{m\!-\!1\over\b} + \Psi(n\!+\!\b) - \Psi(\b\!+\!1)
  ~\to~ -{m\!-\!1\over\b} + \Psi(n) - \Psi(1)
  ~\stackrel{\req{eq:Psi}}\sim~ -{m\!-\!1\over\b} + \ln n
\eeqn
Solving this w.r.t.\ $\b$ we get $\b={m-1\over\ln n}$.
This has net yet the right form
but since $0\leq{\ln m\over\ln n}\leq {m-1\over\ln n}=\b\to 0$, we can write this as
\beqn
  \b ~=~ {m-1\over\ln n} ~\sim~ {m-1\over(1-{\ln m\over\ln n})\ln n} ~=~ {m-1\over\ln{n\over m}}
\eeqn
which apart from the $-1$ is consistent with the $\b$-expressions
in the other regimes.

\section{Proof of Theorem~\ref{prop:redSbstar}}\label{app:redSbstar}

I first prove Theorem~\ref{prop:redSbstar} for $m<n$.
Inserting \req{eq:bstar} into \req{eq:redSbUB} and
abbreviating $\nb:={n\over m}>1$ we get after rearranging terms
\bqan
  R^{\b^*}_S ~\leq~ \Rub^{\b^*}_S
  &=& \CL_w(\A) - m\ln m + \sum_{j\in\A}\fr12\ln{n_j\over 2\pi} + 0.082
  + m\!\cdot\!f(\nb) - \fr12\ln(\nb\ln\nb+1)
\\
  \text{where} & f(\nb) & :=~ \ln\ln\nb + \nb\ln(1+{1\over\nb\ln\nb}) + {\ln(\nb\ln\nb+1)\over\ln\nb}
\eqan
\begin{wrapfigure}{r}{5cm}\vspace{-3ex}
\begin{tikzpicture}
\begin{axis}[
  xlabel=$\nu$,
  title=$f({1\over\nu})-h({1\over\nu})$,
  title style={at={(0.4,0.5)}, anchor=center},
  enlargelimits=false,
  xlabel style={at={(0.9,0.02)}, anchor=south},
  ylabel style={at={(0.00,0.85)}, anchor=north},
  footnotesize,
  xmin=-0.005,xmax=1.005,
  ymin=0,ymax=0.5,
  xtick = {0,0.2,0.4,0.6,0.8,1},
  ytick = {0,0.1,0.2,0.3,0.4,0.5},
  width=5.7cm, height=4cm,
  legend style={ at={(1,0.4)}, anchor=east, draw=none },
  line width=1pt,
]
\addplot[color=blue,mark=none] table {
0.00  0.0
1e-9  0.099146311
1e-6  0.192178430
1e-4  .246565347
0.002  .305700769
0.01	0.352495843
0.03	0.393294354
0.05	0.414790955
0.1	0.445134043
0.15	0.461575924
0.2	0.470868151
0.25	0.475308283
0.3	0.475946058
0.35	0.473320538
0.4	0.467709435
0.45	0.459230506
0.5	0.447884916
0.55	0.433572171
0.6	0.41608659
0.65	0.395096669
0.7	0.370102032
0.75	0.340352807
0.8	0.304694522
0.85	0.26124013
0.9	0.206551459
0.95	0.132862865
0.97 0.093521881
0.99 0.041525051
1.00 0.000000000
};
\addplot[color=red,mark=none,style=dotted,samples=10,domain=0:1] {0.48};
\end{axis}
\end{tikzpicture}
\end{wrapfigure}
It is easy to see that $f(\nb)\sim\ln\ln\nb$ for $\nb\to\infty$ and
$f(1^+)=1$. This motivates the approximation
$h(\nb):=1+\ln(1\!+\!\ln\nb)$, which has the correct $\nb\to 1$
limit and correct $\nb\to\infty$ asymptotics. Next I upper bound
$f(\nb)-h(\nb)$. Since $f-h$ is continuous and tends to zero at 0
and at 1, it is upper bounded by some finite constant. It is easy
to see graphically and numerically but quite cumbersome to show
analytically that $f({1\over\nu})-h({1\over\nu})$ is concave for
$0<\nu<1$ with maximum 0.476... at $\nu=0.284...$,
hence $f(\nb)\leq 1.48+\ln(1\!+\!\ln\nb)$.
Now using $\fr12\ln{n_j\over 2\pi}\,\dot=\,\fr12\ln n_j -0.92$ and
$-\fr12\ln(\nb\ln\nb+1)\leq-\fr12\ln\nb$ (use $\ln x\geq 1-1/x$ on
the inner $\ln\nb$) leads to the desired bound \req{eq:redSbstar} for $m<n$.

For $m=n$, we have $n_i=1\forall i$, hence
$P^{\smash{\v{\hat\t}}}_{iid}=({1\over n})^n$ from \req{IIDjoint},
and $\b^*=\infty$, hence $S(x_{t+1}=i|x_{1:t})=w^t_i$ from
\req{eq:Sdef}, so $S(x_{1:n})=\CL(\A)$. Inserting this into
\req{eq:redDef} gives $R^\infty_S(x_{1:n}) = \CL(\A) - n\ln n$. On
the other hand, \req{eq:redSbstar} for $m=n$ is $\CL(\A)-n\ln
n+0.56n+0.082$, which is clearly larger. \qed

1.48 is a quite crude upper bound on $f(1^+)=1$. By introducing
ugly other terms, one can improve 1.48 to 1 and hence $0.56m$ to
$0.082m$ in bound \req{eq:redSbstar}.

\section{Proof of Theorem~\ref{prop:redSbt}}\label{app:redSbt}

\paradot{$S$ and $R$ for variable $\vec\b$}
For variable $\vec\b$ the joint $S$ distribution and its redundancy are
\bqa
  S^{\vec\b}(x_{1:n}) &=& \prod_{t=0}^{n-1} S^{\b_t}(x_{t+1}|x_{1:t})
  ~=~ \prod_{t=0}^{n-1}{1\over t+\b_t} \prod_{t\in\oldx}n^t_{x_{t+1}} \prod_{t\in\newx}\b_t w^t_{x_{t+1}}
\nonumber\\ \label{eq:RbvItoV}
  R^{\vec\b}_S &=& \underbrace{\phantom{\sum_j}\nq\!\!\!\CL_w(\A)}_{(I)}
  \underbrace{+ \sum_{t=1}^{n-1}\ln(t+\b_t)}_{(II)}
  \underbrace{- \sum_{\nq t\in\newx\setminus\{0\}\nq\nq\nq}\ln\b_t}_{(III)}
  \underbrace{+ \sum_{j\in\A}\ln{n_j^{n_j}\over\G(n_j)}}_{(IV)}
  \underbrace{- \phantom{\sum_j}\nq\!\!\! n\ln n}_{(V)}
\eqa
In the redundancy I removed the $\ln(0+\b_0)-\ln(\b_0)$
contribution. Note that $S^{\vec\b}$ and $R^{\vec\b}_S$ are now not
only dependent on the counts but also on exactly when new symbols
appear, i.e.\ on the $\newx$ set. (for $\vbs=m_t/\ln{t+1\over m_t}$
the dependence is in a sense mild though). The sums cannot be
represented as Gamma functions anymore.

(I) and (IV) and (V) are independent of $\newx$, for (I) by assumption.
(III) obviously depends on $\newx$ but also (II) via $m_t$ in $\b_t$.

\paradot{Redundancy change when swapping two consecutive symbols}
I first show that the earlier new symbols appear, the larger is
$R^\vbs_S$. This fact heavily relies on the specific form of
$\vbs$, which makes the proof cumbersome. Assume at time $t$ there
is an old symbol but at time $t+1$ there is a new symbol for some
$t\in\{1...n-1\}$. That is, $t\in\oldx$ and $m_{t-1}=m_t$ but
$t+1\in\newx$ and $m_{t+1}=m_t+1$. Note that $m_t<t$, and
$x_{t+1}\neq x_t$, since $x_t$ is old and $x_{t+1}$ is new. I now
swap $x_t$ with $x_{t+1}$. I mark all quantities that change by a
prime $'$. That is, $x'_t=x_{t+1}$ and $x'_{t+1}=x_t$. Now $x_t$ is
new ($t-1\in\newx'$) and $x_{t+1}$ is old ($t\in\oldx'$). Further
$m'_t=m_t+1$, and $\b'^{*}_t=m'_t/\ln{t+1\over m'_t}$. Quantities
for all other $t$ remain unchanged. Only one term in (II) and one
term in (III) are affected. The change in redundancy is therefore
\bqan
  \Delta R(t,m_t) &:=& R^{\vec\b'^*}_S - R^\vbs_S
  ~=~ \ln(t\!+\!\b'^*_t) - \ln(t\!+\!\b^*_t) -\ln\b'^*_{t-1} + \ln\b^*_t
\\
  &=& \ln(t\!+\!{m_t+1\over\ln{t+1\over m_t+1}})
    - \ln(t\!+\!{m_t\over\ln{t+1\over m_t}})
    - \ln{m_t\over\ln{t\over m_t}}
    + \ln{m_t\over\ln{t+1\over m_t}}
\eqan
where I have used $m'_{t-1}=m_{t-1}=m_t$.
Collecting terms we get
\beqn
  \Delta R(t,m) ~=~ \ln{\ln{t\over m} + {m+1\over t}{\ln{t\over m}\over\ln{t+1\over m}}
                        \over \ln{t+1\over m}+{m\over t}} ~\stackrel?>~ 0 \qmbox{for} 0<m<t
\eeqn
This is positive, if the numerator is larger than the denominator.
Rearranging terms we can write this as
\beqn
  f_{t,m}(0)\stackrel?>f_{t,m}(1),\qmbox{with}
  f_{t,m}(a) ~:=~ \ln{t+a\over m} - {m+a\over t}{\ln{t\over m}\over\ln{t+a\over m+a}}
\eeqn
Another change in variables gives us
\beqn
  f_{\nb}(x) := f_{t,m}(a) = \ln[\nb(1\!+\!x)] - ({1\over\nb}\!+\!x){\ln\nb\over\ln{1+x\over{1/\nb+x}}},
  \qmbox{where} a=x\!\cdot\!t \mbox{ and } \nb:={t\over m}>1
\eeqn
By differentiation one can show that $f_{\nb}(x)$ is a
decreasing function in $x$ for all $x>0$ and $\nb>1$, which
implies $f_{t,m}(0)>f_{t,m}(1)$ and hence $\Delta R(t,m)>0$.

\paradot{Bounding the redundancy for all new symbols first}
We can repeat swapping symbols and thereby increasing
$R^{\vec\b}_S$ until all symbols appear first before they repeat,
that is, $m_t=\min\{t,m\}$ and $\newx=\{0,...,m-1\}$. For this oder
we have
\beqn
  \b^*_t={t\over\ln{t+1\over t}}\geq t^2 \qmbox{for} t\leq m,
  ~~~~\qmbox{and}~~~~ b^*_t={m\over\ln{t+1\over m}} \qmbox{for} t\geq m
\eeqn
I now bound each of the 5 terms (I)-(V) in $R^{\vec\b}_S$, where I split
the sum in (II) and merge in (III).
\beqn
  \text{(I)} =\underline{\CL_w(\A)} \qmbox{and} \text{(V)} = \underline{-n\ln n} \qquad \text{[nothing to do here]}
\eeqn
\beqn
  \text{(IIa)+(III)} ~=~ \sum_{t=1}^{m-1}\ln(t\!+\!\b^*_t) - \sum_{t=1}^{m-1}\ln\b^*_t
  ~=~ \sum_{t=1}^{m-1} \ln(1\!+\!{t\over\b^*_t})
  ~\leq~ \sum_{t=1}^{m-1} {t\over\b^*_t}
  ~\leq~ \sum_{t=1}^{m-1} {1\over t} ~\leq~ \underline{1+\ln m}
\eeqn
\beqn
  \text{(IIb)} ~=~ \sum_{t=m}^{n-1}\ln(t\!+\!\b^*_t)
  ~=~ \sum_{t=m}^{n-1} \ln t + \sum_{t=m}^{n-1}\ln(1+{m/t\over\ln{t+1\over m}})
\eeqn
Using \req{eq:GammaULB} and \req{eq:lntsum}, the first terms can be bound by
\bqan
  \text{(IIb1)} &=& \sum_{t=m}^{n-1} \ln t ~=~ \ln\G(n)-\ln\G(m) \\
       &\leq& \underline{(n\!-\!\fr12)\ln n - n + 1 - (m\!-\!\fr12)\ln m + m - \ln\sqrt{2\pi}}
\eqan
I split the second term in (IIb) further into $t<2m$ and $t\geq 2m$:
\bqan
  & & \nq \text{(IIb2)} ~=~ \sum_{t=m}^{\nq\nq\min\{2m-1,n-1\}\nq\nq} \ln(1+{m/t\over\ln{t+1\over m}})
  ~\mathop{\leq}\limits^{\scriptstyle\nq\nq\nq\ln{t+1\over m}\geq 1-{m\over t+1}\nq\nq\nq\atop\textstyle\rule{0ex}{2.5ex}\downarrow}~
  \sum_{t=m}^{2m-1} \ln(1+\overbrace{(t\!+\!1)m\over t(t\!+\!1\!-\!m)}^{>1})
  ~\leq~ \sum_{t=m}^{2m-1} \ln{2m(t\!+\!1)\over t(t\!+\!1\!-\!m)}
\\
  & & \nq ~=~ m\ln(2m) + \ln\fr{2m}m - \ln m!
  ~\leq~ \underline{m\ln(2m) + \ln 2 -(m+\fr12)\ln m + m - \ln\sqrt{2\pi}}
\eqan
If $2m<n$
\beqn
  \text{(IIb3)} ~=~ \sum_{t=2m}^{n-1} \ln(1+{m/t\over\ln{t+1\over m}})
  ~\leq~ \sum_{t=2m}^{n-1} {m/t\over\ln{t+1\over m}}
  ~\leq~ {3\over 2}\sum_{t=2m}^{n-1} {1\over{t+1\over m}\ln{t+1\over m}}
\eeqn
where I have used ${t+1\over t}=1+1/t\leq 1+1/2m\leq 3/2$.
If we upper bound the sum by an integral and set $x={t+1\over m}$, we get
\beqn
  \leq~ {3\over 2}\int_{2m-1}^{n-1} {\d t\over{t+1\over m}\ln{t+1\over m}}
  ~=~ {3\over 2}\int_2^{n/m} {m\d x\over x\ln x}
  ~=~ {3\over 2}m[\ln\ln{n\over m} \!-\! \ln\ln2]
  ~\leq~ \underline{{3\over 2}m[\ln\ln{2n\over m} \!-\! \ln\ln2]}
\eeqn
If $2m\geq n$, (IIb3)=0. We can stich both cases together by either
using a $\max$-operation, or as I have done by increasing
$n\leadsto 2n$, which ensures that the last expression is never negative.
\beqn
  \text{(IV)} ~=~ \sum_{j\in\A} \ln{n_j^{n_j}\over\G(n_j)}
  ~\leq~ \sum_{j\in\A} [\fr12\ln n_j+n_j-\ln\sqrt{2\pi}]
  ~=~ \underline{n-{m\over 2}\ln(2\pi) + \sum_{j\in\A} \fr12\ln n_j}
\eeqn

\paradot{Putting everything together}
We can now collect all underlined terms together and get
\bqan
  R^\vbs_S(x_{1:n}) &\leq&
  \textstyle \CL_w(\A) - (m\!-\!1)\ln m + \sum_{j\in\A}\fr12\ln n_j - \fr12\ln n \\
  &+& \fr32 m\ln\ln\fr{2n}m + m[2\!+\!\ln 2\!-\!\fr32\ln\ln 2\!-\!\ln\sqrt{2\pi}] + [2-\ln\pi]
\eqan
Since the all-new-symbols-first order has maximal redundancy,
the bound holds in general.

\section{Improvement on \boldmath$\b^*$}\label{app:bimpr}

Here I generalize $\b^*$ to $\b^c:=\b^*/c$.
From Appendix~\ref{app:bopt} we know that for $n\to\infty$,
the exact optimal $\b^c$ has $1\leq c(\nu)\leq 2$.

\begin{wrapfigure}{r}{5cm}\small
\vspace*{-4ex}\begin{tikzpicture}
\begin{axis}[
  xlabel=$\nu$, ylabel=$c(\nu)$,
  enlargelimits=false,
  xlabel style={at={(0.93,0.02)}, anchor=south},
  ylabel style={at={(0.00,0.85)}, anchor=north},
  footnotesize,
  xmin=-0.005,xmax=1.005,
  ymin=0.995,ymax=2.005,
  xtick = {0,0.2,0.4,0.6,0.8,1},
  width=5.7cm,
  legend style={ at={(1,0.4)}, anchor=east, draw=none, fill=none },
  line width=1pt,
]
\addplot[color=blue,mark=none] table {
1.000000000	2.000000000
0.95532674	1.984994731
0.902446288	1.966918433
0.856951318	1.951066152
0.792103884	1.927931446
0.750568599	1.912740659
0.702006612	1.894563162
0.645880355	1.87292144
0.581976707	1.847319575
0.510632163	1.817277554
0.433015504	1.782386181
0.351424616	1.742387075
0.269485224	1.697277397
0.192046875	1.647427761
0.124518775	1.593682561
7.15E-02	1.537390672
3.52E-02	1.480308797
1.43E-02	1.424356044
4.57E-03	1.371279802
1.09E-03	1.322367006
1.80E-04	1.278322135
1.91E-05	1.239328466
1.19E-06	1.20520869
5.44E-10	1.150005647
4.67E-15	1.109096746
3.45E-34	1.057105132
0.0000000	1.000000000
};
\addplot[color=red,mark=none,style=loosely dotted,samples=50,domain=0.001:1]
{1+x^0.27};
\addplot[color=green,mark=none,style=dashed,samples=50,domain=0.001:1]
{1+ln(1-ln(x))/ln(1/x)};
\legend{$c_{\rm exact}(\nu)$,$1+\nu^{0.27}$,$1+{\ln(1-\ln(\nu))\over\ln(1/\nu)}$}
\end{axis}
\end{tikzpicture}
\vspace*{-4ex}
\end{wrapfigure}
\paradot{Discussion of $c(\nu)$}
The figure on the right plots the exact function $c(\nu)$
implicitly given by $g({c\over\nu}\ln{1\over\nu})=\nu$.
While it is true that $c\to 1$ for $\nu\to 0$, $\b^1$ only starts
to have lower redundancy than $\b^2$ for very small values of
$\nu$, namely $\nu\lesssim 10^{-2}$. So in practice, $c=2$ should
perform better except for $m\lesssim{n\over 100}$.
We could try to find approximate $c(\nu)$ in various ways,
e.g.\ $c(\nu)=1+\nu^{0.27}$ makes
$|g({c(\nu)\over\nu}\ln{1\over\nu})-\nu|<0.002$.
$c(\nu)=1+{\ln(1-\ln(\nu))\over\ln(1/\nu)}$ is theoretically motivated
by an extra iteration of $g$.

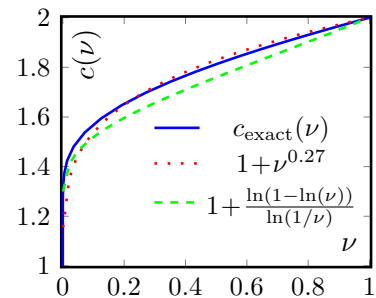
\paradot{Constant $\b^c$}
The proof of Theorem~\ref{prop:redSbstar} in Appendix~\ref{app:redSbstar} still goes through
for $\b^*\leadsto\b^c$ with
now
\beqn
  f_c(\nb) ~=~ \ln(c\ln\nb) + \nb\ln(1+{1\over c\nb\ln\nb}) + {\ln(c\nb\ln\nb+1)\over c\ln\nb}
\eeqn
$f_c-1-\ln(1\!+\!\ln\nb)$ is still upper bounded by 1.48 for all $1\leq c\leq 2$, so
the upper bound in \req{eq:redSbstar} is still valid for
$\b^*\leadsto\b^c$.

\paradot{Variable $\vec\b^c$}
The proof of Theorem~\ref{prop:redSbt} in Appendix~\ref{app:redSbt}
breaks down for $c>1$. $R$ still increases when moving new symbols
earlier if many symbols have already appeared but actually
decreases when only a few symbols have appeared so far. That is,
$\Delta R(t,m)>0$ for large $m$ as before, but $\Delta R(t,m)<0$
for small $m$. $R$ is therefore maximized if all new symbols appear
somewhere in the middle of the sequence. This may lead to a proof
and bound analogous to the $c=1$ case.

Here is a simpler proof with a possibly cruder bound. I reduce
$R^{\smash{\vec\b^{\vec c}}\!\!}_S$ to $R^\vbs_S$ and also allow
for time-dependent $c=c_t$. From expression \req{eq:RbvItoV} it is
easy to see that
\beqn
  R^{\smash{\vec\b^{\vec c}}}_S ~=~ R^\vbs_S + \sum_{t=1}^{n-1} \ln{1+\b^{c_t}_t\over 1+\b^*_t}
  + \sum_{\nq t\in\newx\setminus\{0\}\nq\nq} \ln c_t
  ~\leq~ (m\!-\!1)\ln 2
\eeqn
where I have exploited $c_t\leq 2$ and $\b^{c_t}_t\leq\b^*_t$ for
$c_t\geq 1$. That is, if we add another $(m-1)\ln 2$ to bound
\req{eq:redSbt} it becomes valid for $\vec\b^{\vec c}$ for any
choice of $1\leq c_t\leq 2$.

\section{Bayesian sub-alphabet weighting}\label{app:SAW}

The Bayesian sub-alphabet weighting estimator \cite{Tjalkens:93}
averages over the $\KT_{\A'}$ estimators for all possible
$\A'\subseteq\X$ with a prior uniform in $|\A'|$ and uniform in
$\A'$ given $|\A'|$:
\beq\label{eq:PBayesG}
  P_\Bayes(x_{1:n}) ~=~ \sum_{\nq\A':\A\subseteq\A'\subseteq\X\nq} \text{Prior}(\A')P_{\KT_{\A'}}(x_{1:n})
  \qmbox{with} \text{Prior}(\A') ~=~ {1\over D{D\choose|\A'|}}
\eeq
This mixture of exponential size $2^{D-m}$ can be computed in time
and space linear in $D$ \cite{Tjalkens:93}:
\beq\label{SAWGiMix}
  P_\Bayes(x_{1:n}) ~=~ \sum_{i=1}^D{1\over D}G_i(x_{1:n})
\eeq
with the following sequential representation of $G_i$:
\beqn
  G_i(x_{t+1}|x_{1:t}) ~:=~ \left\{
  \begin{array}{ccl}
    \displaystyle           0                                    & \qmbox{if} & m_{t+1}  >  i \\[1ex]
    \displaystyle        {n^t_{x_{t+1}}+\fr12\over t+i/2}        & \qmbox{if} & m_{t+1}\leq i \qmbox{\&} x_{t+1}\in\A_t \\[1ex]
    \displaystyle {i-m_t\over D-m_t}\!\cdot\!{\frs12\over t+i/2} & \qmbox{if} & m_{t+1}\leq i \qmbox{\&} x_{t+1}\not\in\A_t
  \end{array} \right.
\eeqn
This is still a factor $D$ slower than all other estimators
considered in this paper.

A relation to $S$ can be enforced as follows: First, generalize
$P_{\KT_{\A'}}\equiv P^{\v 1/2}_{\DirM_{\A'}}$ to
$P^{\v\a}_{\DirM_{\A'}}$, then
\beqn
  G^\a_{\b/\a}(x_{t+1}|x_{1:t}) ~\stackrel{\a\to 0}\longrightarrow~ S^\b(x_{t+1}|x_{1:t}) \qmbox{for} w^t_i={1\over D-m_t}
\eeqn
While \req{eq:PBayesG} mixes $G_i$'s, $S^{\b^*}$ maximizes $S^\b$.
So $S^{\b^*}$ with uniform renormalized weights might be an
integer-relaxed, maximum-likelihood approximation of Bayesian
sub-alphabet weighting with Haldane prior. There are several
caveats though.

An anonymous reviewer suggested the following alternative
representation:
\bqan
  P_\Bayes(x_{t+1}|x_{1:t})
  & \propto & P_\Bayes(x_{1:t+1})
  \;=\; \sum_{\nq\A':\A_{t+1}\subseteq\A'\subseteq\X\nq} \text{Prior}(\A')P_{\KT_{\A'}}(x_{1:t+1})
\\
  & = & \sum_{\nq\A':\A_{t+1}\subseteq\A'\subseteq\X\nq} \text{Prior}(\A')
  {\G(\fr12|\A'|)\over\G(t\!+\!1\!+\!\fr12|\A'|)}\prod_{i\in\X}{\G(n^{t+1}_i\!+\!\fr12)\over\G(\fr12)}
\\
  & = & (n^t_{x_{t+1}}\!+\fr12)\bigg(\prod_{i\in\X}{\G(n^t_i\!+\!\fr12)\over\G(\fr12)}\bigg)
  {1\over D}\sum_{\A':\A_{t+1}\subseteq\A'\subseteq\X} {D\choose|\A'|}^{-1}\!\!\!
  {\G(\fr12|\A'|)\over\G(t\!+\!1\!+\!\fr12|\A'|)}
\\
  & \propto & (n^t_{x_{t+1}}\!+\fr12)
  \sum_{m'=m_{t+1}}^D {D-m_{t+1}\choose m'-m_{t+1}}{D\choose m'}^{-1}
  {\G(\fr12 m')\over\G(t\!+\!1\!+\!\fr12 m')}
\eqan
The latter sum can have two values, depending on whether $x_{t+1}$
is new ($m_{t+1}=m_t+1$) or old ($m_{t+1}=m_t$).
We can hence write this as
\bqan
  P_\Bayes(x_{t+1}=i|x_{1:t}) & \propto &
  \left\{ { (n^t_i+\fr12)\g^t_{m_t} \qmbox{if} n^t_i>0, \atop
            ~~~\fr12\g^t_{m_t+1}~~~ \qmbox{if} n^t_i=0, } \right.
\\
  \qmbox{where}~~~
  \g^t_m &:=& \sum_{m'=m}^D
  {D-m\choose m'-m} {D\choose m'}^{-1} {\G(\fr12 m')\over\G(t\!+\!1\!+\!\fr12 m')}
\eqan
By summation, the normalizer can be worked out to be
$(t+\fr12 m_t)\g^t_{m_t} + \fr12 (D-m_t)\g^t_{m_t+1}$,
which allows us to rewrite the result as
\beq\label{SAWbt}
  P_\Bayes(x_{t+1}=i|x_{1:t}) ~=~
  \left\{ { { n^t_i+1/2  \over t+m_t/2+\b_t  } \qmbox{if} n^t_i>0, \atop
            {\b_t/(D-m_t)\over t+m_t/2+\b_t} \qmbox{if} n^t_i=0, } \right.
  \qmbox{with} \b_t:={D- m_t\over 2} {\g^t_{m_t+1}\over\g^t_{m_t}}
\eeq
This has the same structure as \req{eq:Svbs} apart from the $+1/2$
and $+m_t/2$, which is due to using the KT prior rather than a
Haldane prior, and apart from a significantly more complex
expression for $\b_t$, which I expect to be approximately $\b^*_t$.
An advantage of \req{SAWbt} over \req{SAWGiMix} is that not only
can it be used to compute $P_\Bayes(x_{t+1}|x_{1:t})$ in time
$O(D)$ but also the cumulative distribution
$P_\Bayes(X_{t+1}<x_{t+1}|x_{1:t})$, required for arithmetic
coding.

\section{Algorithms \& Applications \& Computation Time}\label{app:Comp}

\def\gSAWBayes{SAW-\Bayes} 
All estimators discussed in this paper, except for Bayesian
sub-alphabet weighting (\gSAWBayes) require just $O(1)$ time and
$O(D)$ space for computing $P(x_{t+1}|x_{1:t})$ and for updating
the relevant parameters like
counts $n_i$, %
the number $m_t$ of symbols seen so far, %
parameter $\beta^*_t$, etc. %
Space can be reduced to $O(m)$ by hashing.
Only \gSAWBayes\ requires $O(D)$ time per $t$ and $O(D)$ space.

Knowledge of $P(x_{t+1}|x_{1:t})$ for all $t$ allows to determine
code length, likelihood, and redundancy of $x_{1:n}$, relevant and
sufficient e.g.\ for model selection such as MDL. Many other tasks
like data compression via arithmetic encoding and Bayesian decision
making require $P(X_{t+1}=i|x_{1:t})$ for all (or at least
multiple) $i\in\X$, which naively requires $O(D)$ time per $t$.

For arithmetic encoding, we actually only need the conditional
distribution function $P(X_{t+1}<x_{t+1}|x_{1:t})$ at $x_{t+1}$ for
$\X\cong\{1,...,D\}$. For $\DirM$ and $S$ this can be computed in
time $O(\log D)$ as follows: Maintain a binary tree of depth
$\lceil\log_2 D\rceil$ with counts $n_1,n_2,...,n_D$ at the leafs
in this order. Inner nodes store the sum of their two children. In
this tree, computing $\sum_{i<x_{t+1}}n_i$ and updating
$n_{x_{t+1}}\leadsto n_{x_{t+1}}+1$ can be performed in time
$O(\log D)$ by accessing/updating the single path from root to leaf
$x_{t+1}$. It is clear how this allows to compute
$\DirM(X_{t+1}<x_{t+1}|x_{1:t})$ in time $O(\log D)$ and space
$O(D)$. Time can be reduced to $O(\log m)$ and space to $O(m)$ by
maintaining a self-balancing binary tree of only the non-zero
counts, which is rebalanced when inserting new non-zero counts.

To compute $S^\vbs(X_{t+1}<x_{t+1}|x_{1:t})$ in time $O(\log D)$,
we have to additionally and in the same way store and maintain
$\tilde w_1, \tilde w_2,...,\tilde w_D$ at the leafs (and their sum
at inner nodes), where $\tilde w_i=w^0_i$ if $i\not\in\A_t$ and
$\tilde w_i=0$ else.

Expectations $\sum_i f(i)P(X_{t+1}=i|x_{1:t})$ can easily be
updated in $O(1)$ time with $O(m)$ space, hence Bayes-optimal
decisions $\arg\min_{y\in\cal Y} \sum_i
\text{Loss}(y,i)P(X_{t+1}=i|x_{1:t})$ can be updated in $O(|{\cal
Y}|)$ time.

A similar tree construction can speed up \gSAWBayes\ \req{SAWGiMix}
from $O(D^2)$ to $O(D\log D)$, or one uses \req{SAWbt}, but time
$O(D)$ seems not further improvable. This renders \gSAWBayes\
impractical for large-alphabet data compression.

Finally, if computation time is at a premium and the logarithm in
$\b^*_t$ too slow, one can with virtually no loss in compression
quality update $1/\ln{t+1\over m_t}$ only whenever $m_t$ or $t$
have changed by more than 10\% since the last update.

\section{Experiments}\label{app:Exp}

\pgfplotsset{ 
  every axis/.append style={
  line width=1pt,
  enlargelimits=false,
 	log basis x=2,
  xlabel style={at={(0.5, 0.08)}},
	ylabel style={at={(0.02,0.8)}, anchor=north},
	}
}
\pgfplotsset{ 
    legend style={
        at={(1.03,1)},
        anchor=north west,
        cells={anchor=west},
        font=\footnotesize,
        rounded corners=2pt,
        line width=1pt
    },
    no markers,
    scale only axis,
}
\pgfplotsset{ 
UsedAlphm/.style={black, dotted},
SbXt/.style={black, solid},
SbXtUB/.style={black, dashdotted},
KTX/.style={green, solid},
Perks/.style={brown, solid},
DirMXt/.style={blue, solid},
SAWBayes/.style={red, solid},
SSDC/.style={yellow, solid},
SbX/.style={black, densely dashed},
SbXUB/.style={black, loosely dashed},
KTACLA/.style={green, densely dashed},
DirMX/.style={blue, densely dashed},
Comb/.style={red, densely dashed},
KTAo/.style={green, loosely dotted},
LLto/.style={red, loosely dotted},
Ho/.style={brown, loosely dotted},
}

\def\gUsedAlphm{$|\A|$ (not a CL)} 
\def\gSbXt{$S^{\vbs\!/2}$} 
\def\gSbXtUB{$\smash{\overline{S}}^\vbs$} 
\def\gKTX{$\KT_\X$} 
\def\gPerks{Perks} 
\def\gDirMXt{D$\vec{\smash{\rm i}}$rM$^*$} 
\def\gSAWBayes{SAW-\Bayes} 
\def\gSSDC{SSDC} 
\def\gSbX{$S^{\b^*\!\!\!/2}$} 
\def\gSbXUB{$\smash{\overline{S}}^{\b^*}$} 
\def\gKTACLA{$\KT_\A\!+\!\ln{D\choose m}$} 
\def\gDirMX{$\DirM^*$} 
\def\gKTAo{$\KT_\A$-{\sc Oracle}} 
\def\gLLto{LL$\v\t$-{\sc Oracle}} 
\def\gHo{H-{\sc Oracle}} 

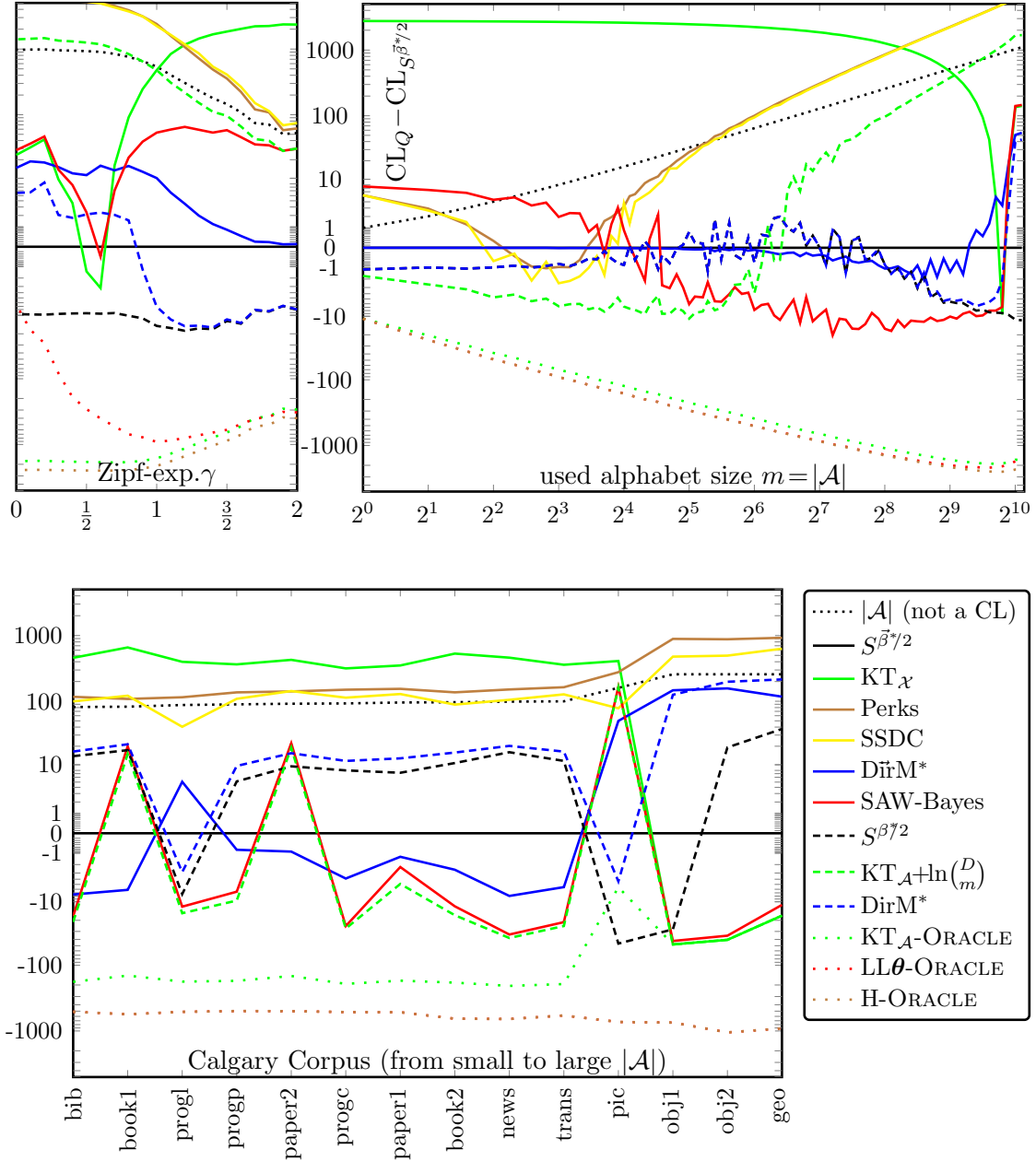
\begin{figure*}
\pgfplotstableread{
     Record.Nr	  Seq.Length.n	    Tot.Alph.s	   Used.Alph.m	         Theta	       zipfexp	          base	          LLto	            Ho	         DirMX	        DirMXt	          SSDC	           SbX	          SbXt	           KTX	          KTAo	        KTACLA	      SAWBayes	         Perks
             1	          1025	         10000	            52	          zipf	           2.0	   2035.200895	   -323.519027	   -395.940749	     -7.920827	      0.098950	     74.253107	     -8.001435	      0.000000	   2335.168105	   -292.634664	     29.809371	     29.509351	     61.194491
             2	          1025	         10000	            50	          zipf	           1.9	   2131.357883	   -317.975311	   -384.756630	     -6.858185	      0.088696	     68.231304	     -6.928862	      0.000000	   2346.961001	   -284.579437	     27.337112	     27.597160	     57.470288
             3	          1025	         10000	            70	          zipf	           1.8	   2581.351648	   -374.358337	   -504.292999	     -8.480348	      0.174989	    113.002706	     -8.622045	      0.000000	   2221.059812	   -374.668542	     39.374179	     33.482561	    107.915124
             4	          1025	         10000	            74	          zipf	           1.7	   2733.948399	   -405.832181	   -527.024937	     -8.048718	      0.193139	    149.714034	     -8.205989	      0.000000	   2197.077945	   -391.856284	     41.865225	     34.601824	    119.755901
             5	          1025	         10000	           108	          zipf	           1.6	   3284.669750	   -487.227967	   -706.232330	    -13.309677	      0.458530	    264.742680	    -13.682077	      0.000000	   2007.332521	   -527.631283	     65.574647	     44.888716	    222.002422
             6	          1025	         10000	           144	          zipf	           1.5	   3784.387723	   -595.543259	   -873.468200	    -11.635390	      0.809785	    406.362628	    -12.321707	      0.000000	   1828.936984	   -654.548159	     95.648755	     57.534654	    354.006854
             7	          1025	         10000	           178	          zipf	           1.4	   4347.326434	   -655.395050	  -1032.487498	    -15.827928	      1.207072	    536.989781	    -16.791057	      0.000000	   1659.015136	   -779.803434	    110.184678	     52.760324	    485.194655
             8	          1025	         10000	           231	          zipf	           1.3	   5035.285987	   -720.652140	  -1243.749464	    -14.699836	      1.964303	    853.517075	    -16.285102	      0.000000	   1431.154499	   -943.868302	    151.204176	     58.386599	    716.517655
             9	          1025	         10000	           304	          zipf	           1.2	   5777.398072	   -790.491047	  -1497.431845	    -14.984236	      3.481136	   1171.427652	    -17.795246	      0.000000	   1154.944437	  -1141.036197	    216.500215	     64.630420	   1064.790672
            10	          1025	         10000	           400	          zipf	           1.1	   6785.795306	   -863.827714	  -1787.897829	    -11.676210	      5.622493	   1630.185588	    -16.085086	      0.000000	    834.570152	  -1368.966278	    306.580731	     58.833308	   1573.292483
            11	          1025	         10000	           537	          zipf	           1.0	   7735.745964	   -909.680462	  -2105.938203	     -7.135929	      9.987811	   2217.967776	    -14.859133	      0.000000	    474.438833	  -1613.531275	    475.122331	     52.703583	   2377.200218
            12	          1025	         10000	           650	          zipf	           0.9	   8449.308611	   -774.076339	  -2300.690104	     -1.227038	     12.692876	   3105.860105	    -11.174646	      0.000000	    244.962244	  -1758.451495	    642.515572	     38.242547	   3108.813987
            13	          1025	         10000	           756	          zipf	           0.8	   8934.636106	   -695.760679	  -2422.853954	      1.511794	     15.750302	   3694.931200	    -10.119750	      0.000000	     90.456410	  -1840.504247	    834.194572	     21.458060	   3848.950527
            14	          1025	         10000	           838	          zipf	           0.7	   9236.383832	   -528.278628	  -2472.292215	      2.039655	     13.376604	   4248.462424	     -9.123919	      0.000000	     16.088686	  -1862.618361	   1012.676629	      6.129599	   4464.568800
            15	          1025	         10000	           899	          zipf	           0.6	   9369.129816	   -400.779470	  -2473.329944	      2.298251	     15.931343	   4878.206748	     -9.345144	      0.000000	     -3.220404	  -1844.929289	   1173.862123	     -0.405505	   4952.016196
            16	          1025	         10000	           923	          zipf	           0.5	   9413.746100	   -288.674058	  -2464.307673	      2.093512	     11.179590	   5222.292293	     -9.186960	      0.000000	     -1.366448	  -1828.921020	   1245.065883	      2.319990	   5153.177474
            17	          1025	         10000	           939	          zipf	           0.4	   9435.347490	   -189.386453	  -2454.565445	      1.745998	     11.973353	   5340.319644	     -9.265678	      0.000000	      3.597801	  -1814.640731	   1295.760385	      7.634375	   5290.673784
            18	          1025	         10000	           950	          zipf	           0.3	   9444.754828	    -81.639610	  -2445.400606	      2.030550	     15.135230	   5490.861151	     -9.422319	      0.000000	      9.514506	  -1802.373688	   1332.887467	     13.739124	   5387.027522
            19	          1025	         10000	           982	          zipf	           0.2	   9449.456389	    -28.018874	  -2404.354454	      8.418215	     17.934921	   5491.624536	     -9.705668	      0.000000	     41.326661	  -1752.333245	   1454.453339	     45.960170	   5678.851694
            20	          1025	         10000	           974	          zipf	           0.1	   9451.972317	    -17.133608	  -2416.914240	      5.680359	     18.686971	   5619.239743	     -9.621462	      0.000000	     31.043486	  -1767.139792	   1421.875123	     35.588947	   5604.038439
            21	          1025	         10000	           966	          zipf	           0.0	   9447.158640	     -6.559758	  -2426.670205	      5.554135	     14.681609	   5491.076168	     -9.512471	      0.000000	     23.666839	  -1779.062114	   1392.108366	     28.115992	   5530.604561
}\loadedtable
\begin{minipage}{0.27\textwidth}
\begin{tikzpicture}[baseline]
\begin{axis}[
  legend to name=leg1,
  footnotesize,
  xlabel={Zipf-exp.$\gamma$},
  xtick={0,0.5,1,1.5,2}, xticklabels={0,$\fr12$,1,$\fr32$,2},
  ymin=-3000, ymax=3000,
  width=\textwidth,
  height=0.32\textheight,
  y coord trafo/.code=\pgfmathparse{(#1>0)-(#1<0))*(ln(abs(#1)+0.01/1.0)-ln(0.01/1.0))},
  ytick={-100,-10,-1,-0.1,-0.01,0,0.01,0.1,1,10,100},
  yticklabels={},
  extra y ticks={-900,-800,...,-200,-90,-80,...,-20,-9,-8,...,-2,-0.9,-0.8,...,-0.1,-0.09,-0.08,...,-0.01,-0.009,-0.008,...,-0.001,0.001,0.002,...,0.009,0.01,0.02,...,0.09,0.1,0.2,...,0.9,2,3,...,9,20,30,...,90,200,300,...,900},
  extra y tick labels={},
  every extra y tick/.style={major tick length=3pt},
  ymin=-50,ymax=50,
  restrict y to domain=-1000:1000,
]
\addplot[UsedAlphm] table [x={zipfexp}, y expr=\thisrow{Used.Alph.m}/100] {\loadedtable}; \addlegendentry{\gUsedAlphm}
\addplot[SbXt]      table [x={zipfexp}, y expr=\thisrow{SbXt}/100]        {\loadedtable}; \addlegendentry{\gSbXt}
\addplot[KTX]       table [x={zipfexp}, y expr=\thisrow{KTX}/100]         {\loadedtable}; \addlegendentry{\gKTX}
\addplot[Perks]     table [x={zipfexp}, y expr=\thisrow{Perks}/100]       {\loadedtable}; \addlegendentry{\gPerks}
\addplot[SSDC]      table [x={zipfexp}, y expr=\thisrow{SSDC}/100]        {\loadedtable}; \addlegendentry{\gSSDC}
\addplot[DirMXt]    table [x={zipfexp}, y expr=\thisrow{DirMXt}/100]      {\loadedtable}; \addlegendentry{\gDirMXt}
\addplot[SAWBayes]  table [x={zipfexp}, y expr=\thisrow{SAWBayes}/100]    {\loadedtable}; \addlegendentry{\gSAWBayes}
\addplot[SbX]       table [x={zipfexp}, y expr=\thisrow{SbX}/100]         {\loadedtable}; \addlegendentry{\gSbX}
\addplot[KTACLA]    table [x={zipfexp}, y expr=\thisrow{KTACLA}/100]      {\loadedtable}; \addlegendentry{\gKTACLA}
\addplot[DirMX]     table [x={zipfexp}, y expr=\thisrow{DirMX}/100]       {\loadedtable}; \addlegendentry{\gDirMX}
\addplot[KTAo]      table [x={zipfexp}, y expr=\thisrow{KTAo}/100]        {\loadedtable}; \addlegendentry{\gKTAo}
\addplot[LLto]      table [x={zipfexp}, y expr=\thisrow{LLto}/100]        {\loadedtable}; \addlegendentry{\gLLto}
\addplot[Ho]        table [x={zipfexp}, y expr=\thisrow{Ho}/100]          {\loadedtable}; \addlegendentry{\gHo}
\end{axis}
\end{tikzpicture}
\end{minipage}
%
\pgfplotstableread{
     Record.Nr	  Seq.Length.n	    Tot.Alph.s	            ms	         Theta	          base	          LLto	            Ho	         DirMX	        DirMXt	          SSDC	           SbX	          SbXt	           KTX	          KTAo	        KTACLA	      SAWBayes	         Perks
             1	          1025	         10000	             1	        random	     10.883777	    -10.883777	    -10.883777	     -1.136102	     -0.000189	      5.259011	     -1.136048	      0.000000	   2736.579771	    -10.883777	     -1.673436	      7.500789	      5.258261
             2	          1025	         10000	             2	        random	    509.465873	    -24.018328	    -24.020448	     -0.991881	     -0.000333	      2.672019	     -0.991771	      0.000000	   2723.096759	    -20.328078	     -2.600644	      6.524165	      2.924720
             3	          1025	         10000	             3	        random	    733.675111	    -36.032015	    -36.035807	     -1.071403	     -0.000869	      1.458431	     -1.071245	      0.000000	   2710.735400	    -29.102065	     -3.263103	      5.801395	      1.261032
             4	          1025	         10000	             4	        random	   1409.525633	    -48.252948	    -48.253910	     -0.983790	     -0.002051	     -0.592001	     -0.983515	      0.000000	   2698.170602	    -38.320812	     -4.658105	      4.336623	      0.270964
             5	          1025	         10000	             5	        random	   1241.933520	    -58.135136	    -58.146126	     -0.871162	     -0.001922	     -0.226287	     -0.870904	      0.000000	   2687.933850	    -45.375170	     -4.111960	      4.804509	     -0.340486
             6	          1025	         10000	             6	        random	   1631.079212	    -69.615926	    -69.621909	     -0.999908	     -0.002191	     -2.144230	     -0.999507	      0.000000	   2676.111172	    -54.138992	     -5.457701	      3.372729	     -0.921945
             7	          1025	         10000	             7	        random	   1657.757489	    -78.921079	    -78.941613	     -0.889021	     -0.001516	     -0.631154	     -0.888669	      0.000000	   2666.448940	    -60.841487	     -4.896366	      3.840788	     -1.042863
             8	          1025	         10000	             8	        random	   2054.316898	    -90.400651	    -90.406706	     -0.757566	     -0.010384	     -2.450846	     -0.756900	      0.000000	   2654.634362	    -69.779044	     -6.703724	      1.933351	     -0.957128
             9	          1025	         10000	             9	        random	   2014.789152	    -99.342719	    -99.350688	     -0.888302	     -0.007352	     -2.184482	     -0.887734	      0.000000	   2645.348587	    -76.258746	     -6.171111	      2.359440	     -0.972304
            10	          1025	         10000	            10	        random	   2197.948684	   -108.476542	   -108.599294	     -0.708859	     -0.008090	     -1.713974	     -0.708277	      0.000000	   2635.770439	    -83.092918	     -6.098429	      2.319454	     -0.535148
            11	          1025	         10000	            11	        random	   2414.463909	   -119.145840	   -119.198329	     -0.462878	     -0.008071	     -1.012550	     -0.461869	      0.000000	   2624.808720	    -91.365862	     -7.559927	      0.739398	      0.095909
            12	          1025	         10000	            12	        random	   2447.943136	   -127.330221	   -127.639018	     -0.494625	     -0.008733	     -0.336205	     -0.493764	      0.000000	   2616.039716	    -97.495765	     -6.965498	      1.209604	      0.564452
            13	          1025	         10000	            13	        random	   2629.345189	   -137.848573	   -137.853443	     -0.120212	     -0.015113	     -1.581273	     -0.118808	      0.000000	   2605.461736	   -105.479794	     -8.305337	     -0.259930	      1.544774
            14	          1025	         10000	            14	        random	   2188.831883	   -144.938146	   -144.975041	     -0.458547	     -0.008509	      0.507160	     -0.457743	      0.000000	   2598.000866	   -110.388090	     -6.643650	      1.266763	      1.909110
            15	          1025	         10000	            15	        random	   2300.608825	   -149.731463	   -151.767714	     -0.715958	     -0.006853	      0.148390	     -0.716000	      0.000000	   2590.909715	   -114.964848	     -4.719519	      3.050757	      2.444306
            16	          1025	         10000	            16	        random	   2570.072078	   -162.148868	   -162.375159	     -0.822870	     -0.010200	      3.713383	     -0.822185	      0.000000	   2579.928892	   -123.466720	     -6.785140	      0.839995	      3.214735
            17	          1025	         10000	            17	        random	   2726.573629	   -171.397259	   -171.444000	      0.067832	     -0.022012	      1.338065	      0.069115	      0.000000	   2570.512674	   -130.437054	     -7.379949	      0.095169	      5.061553
            18	          1025	         10000	            18	        random	   2763.794628	   -179.962562	   -180.200157	     -0.528538	     -0.015308	      3.975723	     -0.527345	      0.000000	   2561.417981	   -137.116850	     -7.741477	     -0.421140	      5.495115
            19	          1025	         10000	            20	        random	   3011.762537	   -198.042477	   -198.128381	      0.260443	     -0.025240	      5.223815	      0.262827	      0.000000	   2542.779215	   -151.011651	     -9.159473	     -2.162564	      8.551285
            20	          1025	         10000	            21	        random	   2884.598519	   -203.462504	   -203.653021	      0.022947	     -0.018667	      5.920862	      0.023744	      0.000000	   2536.957711	   -154.500995	     -6.485001	      0.343446	      9.542375
            21	          1025	         10000	            23	        random	   2857.648341	   -216.090988	   -217.668048	     -1.087022	     -0.013134	      8.170877	     -1.088467	      0.000000	   2522.347700	   -164.519742	     -4.313908	      2.164539	     11.067809
            22	          1025	         10000	            24	        random	   3176.715867	   -230.100149	   -230.197452	     -0.303640	     -0.023791	      9.277439	     -0.301397	      0.000000	   2509.360046	   -175.246021	     -9.010204	     -2.713143	     13.255911
            23	          1025	         10000	            26	        random	   3229.936757	   -244.563631	   -245.939047	     -0.708867	     -0.031168	     14.857403	     -0.706831	      0.000000	   2492.956241	   -187.190646	     -9.016025	     -3.094170	     15.811352
            24	          1025	         10000	            28	        random	   3290.713360	   -257.377555	   -259.510633	     -0.042470	     -0.028075	     14.986349	     -0.041443	      0.000000	   2478.743523	   -197.023262	     -7.061310	     -1.530824	     19.633308
            25	          1025	         10000	            30	        random	   3416.980257	   -275.381193	   -275.790957	      0.105961	     -0.048294	     18.559790	      0.108280	      0.000000	   2461.728482	   -209.731347	     -8.122915	     -2.999522	     23.120955
            26	          1025	         10000	            32	        random	   3607.844185	   -293.368583	   -293.566001	      0.738050	     -0.055069	     22.406817	      0.743811	      0.000000	   2443.229760	   -223.991142	    -10.867862	     -6.166884	     27.264214
            27	          1025	         10000	            34	        random	   3633.863553	   -303.667066	   -306.399279	      0.197663	     -0.061901	     26.065208	      0.201016	      0.000000	   2429.777276	   -233.268278	     -8.753696	     -4.490097	     30.390406
            28	          1025	         10000	            37	        random	   3664.687412	   -327.121736	   -327.869332	      0.157646	     -0.058141	     33.482498	      0.160910	      0.000000	   2407.277894	   -249.615410	     -8.230110	     -4.649928	     36.133493
            29	          1025	         10000	            39	        random	   3650.238779	   -335.605989	   -339.647104	     -1.266172	     -0.019825	     36.110398	     -1.268153	      0.000000	   2394.933420	   -257.927319	     -5.430000	     -2.323233	     38.736370
            30	          1025	         10000	            42	        random	   3867.590932	   -360.949115	   -363.354245	      1.681990	     -0.130556	     43.176837	      1.687825	      0.000000	   2370.080232	   -276.827738	     -7.851543	     -5.480948	     47.984071
            31	          1025	         10000	            45	        random	   3910.508309	   -380.915960	   -383.834913	     -0.498294	     -0.046461	     45.904377	     -0.495944	      0.000000	   2348.642326	   -292.420555	     -7.178319	     -5.574444	     52.371757
            32	          1025	         10000	            49	        random	   3962.901338	   -403.552307	   -408.912021	      1.394874	     -0.086000	     59.046453	      1.396956	      0.000000	   2322.222357	   -311.201762	     -4.578618	     -4.043321	     63.428327
            33	          1025	         10000	            52	        random	   3991.067485	   -421.840540	   -427.511456	     -0.291618	     -0.047102	     65.755172	     -0.295698	      0.000000	   2302.708077	   -325.094691	     -2.650657	     -2.950676	     68.891448
            34	          1025	         10000	            56	        random	   4086.686058	   -450.093950	   -454.231736	     -0.506831	     -0.049697	     71.949127	     -0.510592	      0.000000	   2274.620621	   -345.820136	     -2.548158	     -4.006027	     78.568279
            35	          1025	         10000	            60	        random	   4197.740221	   -476.251064	   -479.248737	      0.433579	     -0.106351	     85.391380	      0.429601	      0.000000	   2248.282515	   -364.937996	     -1.123099	     -3.788103	     89.777766
            36	          1025	         10000	            64	        random	   4340.861058	   -505.969752	   -507.885774	     -0.672122	     -0.075191	     93.619016	     -0.669673	      0.000000	   2218.190507	   -387.941648	     -3.850092	     -7.770432	     99.301340
            37	          1025	         10000	            69	        random	   4413.454286	   -534.749051	   -536.912910	      0.576218	     -0.122939	    108.702505	      0.576315	      0.000000	   2187.545627	   -409.898528	     -0.810728	     -6.366547	    114.293896
            38	          1025	         10000	            74	        random	   4569.223981	   -561.362104	   -566.471534	      0.506670	     -0.159031	    124.298330	      0.503287	      0.000000	   2156.393575	   -432.540654	      1.180855	     -6.082546	    128.454237
            39	          1025	         10000	            79	        random	   4744.026088	   -596.437787	   -599.060708	      1.801040	     -0.212713	    138.819451	      1.812603	      0.000000	   2121.998499	   -458.591272	     -0.574390	     -9.615971	    144.452351
            40	          1025	         10000	            84	        random	   4851.062533	   -626.740016	   -629.453863	      1.956087	     -0.256952	    146.410716	      1.974802	      0.000000	   2089.875462	   -482.524336	     -0.529269	    -11.418284	    159.734357
            41	          1025	         10000	            91	        random	   4879.517567	   -657.958994	   -664.766210	      1.020715	     -0.195566	    172.527333	      1.016681	      0.000000	   2052.447540	   -508.727652	      6.339082	     -7.250414	    180.625256
            42	          1025	         10000	            97	        random	   4972.784297	   -686.954334	   -695.705298	      0.768026	     -0.182735	    192.567499	      0.749430	      0.000000	   2019.654004	   -532.108639	     10.873154	     -5.134974	    199.678750
            43	          1025	         10000	           104	        random	   5185.168923	   -730.141883	   -740.620045	      1.363587	     -0.343538	    212.668855	      1.383440	      0.000000	   1972.162613	   -568.844453	      6.234900	    -12.713939	    223.501090
            44	          1025	         10000	           111	        random	   5152.320735	   -759.377850	   -769.618552	      0.106385	     -0.199459	    238.715910	      0.072022	      0.000000	   1941.181590	   -589.296560	     17.406153	     -4.609350	    246.042403
            45	          1025	         10000	           119	        random	   5375.772576	   -811.471511	   -817.864898	      0.519776	     -0.316401	    266.182755	      0.529281	      0.000000	   1890.045600	   -628.657721	     13.642997	    -12.028908	    274.517501
            46	          1025	         10000	           128	        random	   5421.540630	   -846.250990	   -860.455908	      0.015621	     -0.277626	    296.539876	      0.001833	      0.000000	   1844.596654	   -661.165855	     20.535957	     -9.439596	    306.427650
            47	          1025	         10000	           137	        random	   5500.401935	   -894.845297	   -902.730138	     -0.835906	     -0.329043	    327.233958	     -0.874289	      0.000000	   1799.470803	   -693.651449	     26.812390	     -7.665040	    338.901081
            48	          1025	         10000	           147	        random	   5812.604885	   -955.231796	   -962.411616	      0.846011	     -0.557432	    364.975773	      0.922764	      0.000000	   1736.089808	   -743.314675	     19.518669	    -20.189888	    378.756965
            49	          1025	         10000	           158	        random	   5823.054460	   -990.466942	  -1004.387145	      0.460819	     -0.440936	    410.601922	      0.448397	      0.000000	   1690.888802	   -773.794821	     34.851316	    -10.886093	    421.308207
            50	          1025	         10000	           169	        random	   5923.927207	  -1032.418232	  -1054.407905	     -1.399263	     -0.358628	    449.616084	     -1.431984	      0.000000	   1637.366123	   -812.954887	     40.727731	    -11.322547	    463.593593
            51	          1025	         10000	           181	        random	   6137.470727	  -1094.186318	  -1111.277481	      0.290380	     -0.616087	    506.903996	      0.323624	      0.000000	   1576.450529	   -858.581287	     43.404819	    -15.852403	    514.759236
            52	          1025	         10000	           194	        random	   6251.080225	  -1149.531451	  -1166.628747	      0.412362	     -0.684936	    552.673525	      0.454557	      0.000000	   1516.856661	   -902.026344	     51.377615	    -16.060285	    569.791146
            53	          1025	         10000	           208	        random	   6357.282632	  -1195.810086	  -1221.678084	     -0.515549	     -0.614276	    616.668677	     -0.528496	      0.000000	   1457.481870	   -944.457169	     63.329753	    -13.348168	    629.444371
            54	          1025	         10000	           223	        random	   6418.217543	  -1249.517159	  -1276.875417	     -0.328821	     -0.577229	    679.470440	     -0.419311	      0.000000	   1397.619265	   -986.642290	     78.347545	     -8.721486	    696.169130
            55	          1025	         10000	           239	        random	   6609.105240	  -1305.576905	  -1340.428951	     -0.055700	     -0.847673	    746.410584	     -0.113228	      0.000000	   1328.894753	  -1037.014418	     87.857639	    -10.851732	    769.313374
            56	          1025	         10000	           256	        random	   6798.463330	  -1378.507546	  -1408.338170	     -0.816933	     -1.070169	    832.196802	     -0.788036	      0.000000	   1255.422925	  -1091.514477	     95.783221	    -15.920632	    847.945619
            57	          1025	         10000	           274	        random	   6953.479031	  -1431.870201	  -1473.887643	     -1.485831	     -1.041496	    918.387426	     -1.436474	      0.000000	   1184.157398	  -1143.241251	    108.894432	    -17.270171	    933.320280
            58	          1025	         10000	           294	        random	   7049.642934	  -1481.555754	  -1536.846881	     -0.971005	     -0.915701	   1014.825060	     -1.057723	      0.000000	   1115.084055	  -1191.217928	    131.538810	    -11.539599	   1031.640007
            59	          1025	         10000	           315	        random	   7229.677679	  -1553.916973	  -1609.056504	     -1.921020	     -1.213760	   1123.073822	     -1.932829	      0.000000	   1036.168816	  -1248.624214	    146.778956	    -15.019496	   1136.153692
            60	          1025	         10000	           338	        random	   7352.333589	  -1614.257464	  -1677.751244	     -1.222066	     -1.474220	   1237.843938	     -1.316015	      0.000000	    960.337058	  -1301.596234	    171.718115	    -11.717386	   1255.069139
            61	          1025	         10000	           362	        random	   7511.961141	  -1681.281650	  -1749.571605	     -2.941972	     -0.849078	   1362.305510	     -3.108120	      0.000000	    881.138749	  -1357.659419	    195.229309	    -12.057670	   1379.791271
            62	          1025	         10000	           388	        random	   7687.281853	  -1746.702524	  -1822.912143	     -1.405199	     -2.058249	   1504.515365	     -1.490128	      0.000000	    799.624853	  -1414.874833	    222.360805	    -12.249361	   1521.874208
            63	          1025	         10000	           416	        random	   7815.809922	  -1814.508637	  -1895.650696	     -4.318824	     -0.523167	   1660.396655	     -4.671588	      0.000000	    718.445441	  -1470.707431	    255.340514	    -10.446609	   1673.794775
            64	          1025	         10000	           446	        random	   7998.127345	  -1868.496881	  -1971.559979	     -3.605146	     -1.777859	   1827.619982	     -3.838526	      0.000000	    633.181423	  -1529.685080	    289.340952	    -11.906059	   1845.082045
            65	          1025	         10000	           478	        random	   8142.697226	  -1919.098043	  -2045.060428	     -5.388509	     -0.459501	   2012.574083	     -5.800071	      0.000000	    549.945704	  -1585.797093	    330.081541	    -11.392958	   2029.690567
            66	          1025	         10000	           512	        random	   8297.668738	  -1988.843194	  -2116.135746	     -5.040220	     -1.327665	   2209.010817	     -5.581211	      0.000000	    468.535417	  -1639.341885	    376.981387	    -10.039743	   2233.407978
            67	          1025	         10000	           549	        random	   8470.165415	  -2026.741312	  -2188.629334	     -5.345975	     -0.576259	   2435.961121	     -5.889349	      0.000000	    384.781607	  -1693.798872	    429.133379	    -10.867592	   2460.642471
            68	          1025	         10000	           588	        random	   8616.081056	  -2079.988552	  -2254.478294	     -5.095540	     -1.529673	   2680.872883	     -5.918133	      0.000000	    307.211103	  -1741.555411	    490.895782	     -8.992020	   2707.297934
            69	          1025	         10000	           630	        random	   8763.306684	  -2121.014330	  -2317.678232	     -5.732093	      0.146728	   2955.890860	     -6.912263	      0.000000	    231.440701	  -1786.343621	    560.982541	     -8.327510	   2979.869551
            70	          1025	         10000	           676	        random	   8920.676628	  -2113.460343	  -2378.154721	     -6.697157	      0.841812	   3259.711957	     -7.980927	      0.000000	    157.148994	  -1827.931675	    641.787933	     -9.882786	   3288.027992
            71	          1025	         10000	           724	        random	   9062.480565	  -2159.785278	  -2426.889194	     -6.160585	      0.499773	   3596.355070	     -7.711212	      0.000000	     94.113181	  -1858.098881	    735.759460	     -9.616248	   3622.118562
            72	          1025	         10000	           776	        random	   9185.428366	  -2151.620148	  -2463.267172	     -5.680020	      1.308905	   3962.938874	     -8.500459	      0.000000	     42.492596	  -1875.439383	    849.037339	     -7.987076	   3996.710799
            73	          1025	         10000	           832	        random	   9301.533696	  -2128.880477	  -2484.572041	     -4.588727	      4.873422	   4382.091791	     -8.621416	      0.000000	      4.721899	  -1877.707254	    983.216696	     -7.966737	   4418.635040
            74	          1025	         10000	           891	        random	   9390.431755	  -2056.657948	  -2481.310093	     -2.125496	      3.375026	   4845.943932	     -8.779507	      0.000000	     -9.230718	  -1855.706888	   1144.530485	     -6.902904	   4886.845033
            75	          1025	         10000	           955	        random	   9445.008893	  -1945.339396	  -2439.769696	      3.317167	     18.980278	   5384.073483	     -9.459859	      0.000000	     13.731851	  -1795.284271	   1351.230309	     18.031972	   5431.439464
            76	          1025	         10000	          1024	        random	   9402.876694	  -1725.001153	  -2298.503900	     39.682188	     50.415793	   6059.394167	    -11.275151	      0.000000	    135.069723	  -1635.199033	   1663.715598	    140.070992	   6113.656079
            77	          1025	         10000	          1097	        random	   9398.808822	  -1873.738944	  -2293.049733	     42.787949	     54.370056	   6072.564336	    -11.582107	      0.000000	    140.236208	  -1629.482754	   1671.601739	    145.244796	   6126.934392
}\loadedtable
\hspace{1ex}%
\begin{minipage}{0.7\textwidth}
\begin{tikzpicture}[baseline]
\begin{semilogxaxis}[
  legend to name=leg2,
  footnotesize,
  xlabel={used alphabet size $m=|\A|$},
  xmax=1125,
  ylabel={$\CL_Q-\CL_{\text{\gSbXt}}$},
  ymin=-3000, ymax=3000,
  width=\textwidth - \widthof{-1000},
  height=0.32\textheight,
  y coord trafo/.code=\pgfmathparse{((#1>0)-(#1<0))*(ln(abs(#1)+0.01/1.0)-ln(0.01/1.0))},
  ytick={-100,-10,-1,-0.1,-0.01,0,0.01,0.1,1,10,100},
  yticklabels={-10000,-1000,-100,-10,-1,0,1,10,100,1000,10000},
  yticklabel style={text width=\widthof{-1000}, inner sep=1pt, xshift=0.7ex, align=center},
  extra y ticks={-900,-800,...,-200,-90,-80,...,-20,-9,-8,...,-2,-0.9,-0.8,...,-0.1,-0.09,-0.08,...,-0.01,-0.009,-0.008,...,-0.001,0.001,0.002,...,0.009,0.01,0.02,...,0.09,0.1,0.2,...,0.9,2,3,...,9,20,30,...,90,200,300,...,900},
  extra y tick labels={},
  every extra y tick/.style={major tick length=3pt},
  ymin=-50,ymax=50,
  restrict y to domain=-1000:1000,
]
\addplot[UsedAlphm] table [x={ms}, y expr=\thisrow{ms}/100]          {\loadedtable}; \addlegendentry{\gUsedAlphm}
\addplot[SbXt]      table [x={ms}, y expr=\thisrow{SbXt}/100]        {\loadedtable}; \addlegendentry{\gSbXt}
\addplot[KTX]       table [x={ms}, y expr=\thisrow{KTX}/100]         {\loadedtable}; \addlegendentry{\gKTX}
\addplot[Perks]     table [x={ms}, y expr=\thisrow{Perks}/100]       {\loadedtable}; \addlegendentry{\gPerks}
\addplot[SSDC]      table [x={ms}, y expr=\thisrow{SSDC}/100]        {\loadedtable}; \addlegendentry{\gSSDC}
\addplot[DirMXt]    table [x={ms}, y expr=\thisrow{DirMXt}/100]      {\loadedtable}; \addlegendentry{\gDirMXt}
\addplot[SAWBayes]  table [x={ms}, y expr=\thisrow{SAWBayes}/100]    {\loadedtable}; \addlegendentry{\gSAWBayes}
\addplot[SbX]       table [x={ms}, y expr=\thisrow{SbX}/100]         {\loadedtable}; \addlegendentry{\gSbX}
\addplot[KTACLA]    table [x={ms}, y expr=\thisrow{KTACLA}/100]      {\loadedtable}; \addlegendentry{\gKTACLA}
\addplot[DirMX]     table [x={ms}, y expr=\thisrow{DirMX}/100]       {\loadedtable}; \addlegendentry{\gDirMX}
\addplot[KTAo]      table [x={ms}, y expr=\thisrow{KTAo}/100]        {\loadedtable}; \addlegendentry{\gKTAo}
\addplot[LLto]      table [x={ms}, y expr=\thisrow{LLto}/100]        {\loadedtable}; \addlegendentry{\gLLto}
\addplot[Ho]        table [x={ms}, y expr=\thisrow{Ho}/100]          {\loadedtable}; \addlegendentry{\gHo}
\end{semilogxaxis}
\end{tikzpicture}
\end{minipage}

\vspace{3ex}
\pgfplotstableread{
     NiceOrder     Record.Nr	  Seq.Length.n	    Tot.Alph.s	   Used.Alph.m	         Theta	      filename	          base	          LLto	            Ho	         DirMX	        DirMXt	          SSDC	           SbX	          SbXt	           KTX	          KTAo	        KTACLA	      SAWBayes	         Perks
             1             1	        111261	           256	            81	          file	           bib	 401584.648676	   -507.200297	   -507.200297	     16.341412	     -7.477210	     98.823090	     13.756163	      0.000000	    455.773353	   -177.930092	    -21.079740	    -17.116901	    116.151287
             2             2	        768771	           256	            82	          file	         book1	2412942.207284	   -553.982178	   -553.982178	     21.330720	     -6.217114	    120.653783	     17.186011	      0.000000	    655.113585	   -142.734824	     14.873594	     19.811136	    107.717473
             3            12	         71646	           256	            87	          file	         progl	 237396.657435	   -508.398372	   -508.398372	     -2.880266	      5.033692	     40.144717	     -7.373424	      0.000000	    396.420380	   -176.443352	    -15.310194	    -11.967860	    114.300166
             4            13	         49379	           256	            89	          file	         progp	 167137.374946	   -494.324159	   -494.324159	      9.482922	     -0.780126	    108.240504	      5.108808	      0.000000	    362.520310	   -171.901368	     -9.480320	     -6.691316	    135.773923
             5             9	         82199	           256	            91	          file	        paper2	 262667.851659	   -496.484268	   -496.484268	     15.364411	     -0.896960	    142.181418	      9.343806	      0.000000	    424.608769	   -144.786244	     18.854116	     22.201948	    140.356216
             6            11	         39611	           256	            92	          file	         progc	 143258.278905	   -512.779525	   -512.779525	     11.486824	     -3.857504	    112.651896	      7.974503	      0.000000	    314.825509	   -191.046529	    -26.822012	    -24.520144	    149.095998
             7             8	         53161	           256	            95	          file	        paper1	 184130.446065	   -515.526684	   -515.526684	     12.667718	     -1.273323	    128.061898	      7.262244	      0.000000	    348.636696	   -170.743815	     -4.867856	     -2.250252	    153.647994
             8             3	        610856	           256	            96	          file	         book2	2029910.156499	   -646.575774	   -646.575774	     15.735733	     -2.572521	     88.235559	     10.600625	      0.000000	    528.190765	   -182.879291	    -16.486276	    -11.804470	    135.894078
             9             5	        377109	           256	            98	          file	          news	1357182.560683	   -654.165632	   -654.165632	     20.103774	     -7.920020	    104.451951	     15.934677	      0.000000	    458.423369	   -205.098413	    -37.720999	    -33.311171	    150.356032
            10            14	         93695	           256	            99	          file	         trans	 359900.687682	   -577.407858	   -577.407858	     16.299603	     -5.556326	    126.628892	     11.495639	      0.000000	    357.469410	   -192.320634	    -24.475744	    -21.263469	    162.614376
            11            10	        513216	           256	           159	          file	           pic	 431225.706079	   -724.707383	   -724.707383	     -4.361377	     49.507171	     77.487401	    -46.034724	      0.000000	    407.694022	     -5.567004	    161.326474	    164.864666	    275.779714
            12             6	         21504	           256	           256	          file	          obj1	  89401.166556	   -741.085922	   -741.085922	    124.728236	    146.022164	    475.205506	    -27.590286	      0.000000	    -47.501909	    -47.501909	    -47.501909	    -41.956731	    883.687744
            13             7	        246814	           256	           256	          file	          obj2	1072060.271824	  -1044.147147	  -1044.147147	    197.080159	    156.127197	    491.925691	     19.001535	      0.000000	    -40.268107	    -40.268107	    -40.268107	    -34.722930	    871.943569
            14             4	        102400	           256	           256	          file	           geo	 401678.474400	   -908.483600	   -908.483600	    213.922101	    116.784191	    627.273331	     37.117750	      0.000000	    -16.732470	    -16.732470	    -16.732470	    -11.187292	    918.591762
}\loadedtable
\begin{minipage}{\textwidth}
\begin{tikzpicture}[baseline]
\begin{axis}[
  footnotesize,
  xlabel={Calgary Corpus (from small to large $|\A|$)},
  xtick={1,...,18},
  xticklabel style={rotate=90,anchor=east},
  xticklabels from table={\loadedtable}{filename},
  ymin=-3000, ymax=3000,
  width=\textwidth - \widthof{-1000} - 21ex,
  height=0.32\textheight,
  y coord trafo/.code=\pgfmathparse{(#1>0)-(#1<0))*(ln(abs(#1)+0.01/1.0)-ln(0.01/1.0))},
  ytick={-100,-10,-1,-0.1,-0.01,0,0.01,0.1,1,10,100},
  yticklabels={-10000,-1000,-100,-10,-1,0,1,10,100,1000,10000},
  extra y ticks={-900,-800,...,-200,-90,-80,...,-20,-9,-8,...,-2,-0.9,-0.8,...,-0.1,-0.09,-0.08,...,-0.01,-0.009,-0.008,...,-0.001,0.001,0.002,...,0.009,0.01,0.02,...,0.09,0.1,0.2,...,0.9,2,3,...,9,20,30,...,90,200,300,...,900},
  extra y tick labels={},
  every extra y tick/.style={major tick length=3pt},
  ymin=-50,ymax=50,
  restrict y to domain=-1000:1000,
]
\addplot[UsedAlphm] table [x={NiceOrder}, y expr=\thisrow{Used.Alph.m}/100] {\loadedtable}; \addlegendentry{\gUsedAlphm}
\addplot[SbXt]      table [x={NiceOrder}, y expr=\thisrow{SbXt}/100]        {\loadedtable}; \addlegendentry{\gSbXt}
\addplot[KTX]       table [x={NiceOrder}, y expr=\thisrow{KTX}/100]         {\loadedtable}; \addlegendentry{\gKTX}
\addplot[Perks]     table [x={NiceOrder}, y expr=\thisrow{Perks}/100]       {\loadedtable}; \addlegendentry{\gPerks}
\addplot[SSDC]      table [x={NiceOrder}, y expr=\thisrow{SSDC}/100]        {\loadedtable}; \addlegendentry{\gSSDC}
\addplot[DirMXt]    table [x={NiceOrder}, y expr=\thisrow{DirMXt}/100]      {\loadedtable}; \addlegendentry{\gDirMXt}
\addplot[SAWBayes]  table [x={NiceOrder}, y expr=\thisrow{SAWBayes}/100]    {\loadedtable}; \addlegendentry{\gSAWBayes}
\addplot[SbX]       table [x={NiceOrder}, y expr=\thisrow{SbX}/100]         {\loadedtable}; \addlegendentry{\gSbX}
\addplot[KTACLA]    table [x={NiceOrder}, y expr=\thisrow{KTACLA}/100]      {\loadedtable}; \addlegendentry{\gKTACLA}
\addplot[DirMX]     table [x={NiceOrder}, y expr=\thisrow{DirMX}/100]       {\loadedtable}; \addlegendentry{\gDirMX}
\addplot[KTAo]      table [x={NiceOrder}, y expr=\thisrow{KTAo}/100]        {\loadedtable}; \addlegendentry{\gKTAo}
\addplot[LLto]      table [x={NiceOrder}, y expr=\thisrow{LLto}/100]        {\loadedtable}; \addlegendentry{\gLLto}
\addplot[Ho]        table [x={NiceOrder}, y expr=\thisrow{Ho}/100]          {\loadedtable}; \addlegendentry{\gHo}
\end{axis}
\end{tikzpicture}
\end{minipage}
\caption{\label{fig:results}Plotted are code length differences to
\gSbXt\ of various estimators. The two top graphs are for fixed
sequence length $n=1024$ and total alphabet size $D=10\,000$ for
varying Zipf exponents $\gamma$ and used alphabet sizes $m=|\A|$.
The bottom graph is for the 14 files from the Calgary corpus with
$21504\leq n\leq 768\,771$ and byte alphabet ($D=256$). The
online/offline/oracle estimators have solid/dashed/dotted lines. A
curve above/below zero means worse/better than \gSbXt. The black
dotted curve is not a code length but shows the used alphabet size
$m\equiv|\A|$. }
\end{figure*}

I determined the code length of various estimators for various
sequence lengths $n$, used alphabet sizes $m$, and base alphabet
sizes $D$ on artificially generated data sequences and the Calgary
corpus.
I consider the new estimator $S$ and the Dirichlet-multinomial with
approximately optimal constant $\b^*\!\!/2$ and variable
$\vbs\!/2$ and with Perks prior, the KT estimator for the base and
for the used alphabet, and Bayesian sub-alphabet weighting,
introduced in Section~\ref{sec:Compare}. I also compare against the
true distribution and the empirical entropy.

\paradot{Data generation}
I sampled $\t_1,...,\t_m$ uniformly from the $m-1$-dimensional
probability simplex and set $\t_{m+1}=...=\t_D=0$. I then sampled
$x_{1:n}$ from $P^{\v\t}_{iid}$. Unless $n\gg m$ or $D\gg n$, this
usually results in sequences that actually contain less than $m$
symbols, and e.g.\ $|\A|=n$ is virtually impossible to achieve in
this way.
%
I therefore generate sequences by first setting $x_t=t$ for
$t=1...\min\{m,n\}$, then sample the remaining $x_t$ from
$P^{\v\t}_{iid}$, and then scramble the result. The resulting code
lengths were virtually indistinguishable from the ``normal''
i.i.d.\ sampling, when the latter was also feasible.

I also generated sequences with a version of D'Hondt's method for
allocating seats in party-list proportional representation, which
ensures $|n_i-\t_i\cdot n|<1$ and adapted it to also ensure $n_i>0$
if $\t_i>0$ and $i\leq n$ by dividing by zero (rather than 1)
first. As expected, the results were a bit less noisy, but
otherwise very similar.

In another experiment I chose $\v\t$ to be Zipf-distributed, i.e.\
$\t_i\propto i^{-\gamma}$ with varying Zipf exponent $\gamma>0$,
which for $\gamma\approx 1$ mimics quite well the empirical
distribution of words in English texts. The larger $\gamma$, the
smaller the used alphabet $\A$.

\paradot{My $S$-estimators}
I determined the code length of my models (\gSbX and \gSbXt) with
constant and variable optimal $\b^*$. I chose uniform normalized
weights $w^t_i=1/(D-m_t)$.
I also played around with other $\b$ and $\vec\b$, but performance
either severely deteriorated, or only marginally and locally
improved. The code length is very sensitive to some changes, e.g.\
$\b=m/\ln n$ and $\b=m/\ln{2n\over m}$ perform badly for large $m$,
since these $\b$ have the wrong scaling for $m\to n$, but less
sensitive to other changes, e.g.\ $\b=(m+c)/\ln{n+c'\over m+c}$ for
small $c,c'$ are generally ok. For the experiments I used
$\b^{c=2}=\b^*\!\!/2$ and $\b^{c=2}_t=\b^*_t/2$.

\paradot{Other estimators}
I also determined the code length of the other estimators
discussed in Section~\ref{sec:Compare}.
I considered: \\
(i) the Dirichlet-multinomial with $\v\a=\v 1/D$ (\gPerks) and
optimized constant $\v\a^*$ (\gDirMX) and optimal variable
$\vec{\v\a}^*$ (\gDirMXt) with uniform weights \req{eq:astar};
\\
(ii) the \KT-estimator with base alphabet $\X$ (\gKTX), %
\\
(iii) the \KT-estimator for used alphabet $\A$ (\gKTAo), %
a feasible off-line version by pre-coding $\A$ (\gKTACLA), %
and the online version using escape probability $\frs{1}{t+1}$
(\gSSDC) discussed in Section~\ref{sec:Compare};
\\
(iv) the Bayesian sub-alphabet weighting (\gSAWBayes) discussed in
Appendix~\ref{app:SAW};
\\
(v) the empirical entropy $n H({\v n\over n})=\sum_i n_i\ln{n\over
n_i}$ (\gHo);
\\
(vi) the log-likelihood of the sampling distribution $\ln
1/P^{\v\t}_{iid}$ (\gLLto) for artificial data.

\paradot{Results}
Figure \ref{fig:results} plots the results for the various
estimators. The vertical axis is the code length (or redundancy)
difference of the estimator under consideration and our prime model
\gSbXt. So negative/positive values indicate better/worse
performance than \gSbXt.
The two top graphs are for artificially generated data with fixed
sequence length $n=1024$ and total alphabet size $D=10\,000$. In
the right graph I varied $m=1,2,4,...2^{10}$ and in the left graph
I varied the Zipf exponent $\gamma\in[0;2]$. The bottom graph shows
results for the 14 files from the Calgary corpus with byte alphabet
($D=256$).
All results are plotted and discussed relative to \gSbXt.
Rather than averaging over multiple runs and plotting error bars
for the artificial data, I generated (necessarily) one new sequence
for each $\gamma$ and $m$ for sufficiently many $\gamma$ and $m$.
The noise level of the curves captures the sample variation very
well.

\paradot{Discussion}
The results generally confirm the theory with few/small surprises.

The online estimators are plotted with solid lines. \gDirMXt\
mostly coincides within $\pm 10$ nits with \gSbXt\ for most $m$.
Only when $m$ approached $n$ is \gSbXt\ superior to \gDirMXt\ due
to renormalized weights leading to shorter $\CL_w(\A)$.
Among the proper estimators, \gSAWBayes\ works best by a small
margin, except for very small ($m\lesssim\ln n$) and very large
($m\approx n$) used alphabet and Zipf distributed data, but note
that it is $D$ (here $10\,000$ or 256) times slower than all the
other algorithms. \gSSDC\ is virtually indistinguishable from
\gPerks\ on the artificial data and only slightly better on the
real data. Both perform poorly except for very small $m\lesssim\ln
n$. Note that \gPerks\ performs as well as \gDirMX\ (only) around
$m\approx 2\ln{n\over m}$, i.e.\ when their priors coincide.
\gKTX\ as well as $\DirM^{\v\a}$ with any other fixed choice of
$\v\a$ perform very badly, especially for small $m$. \gKTX\
performs well only for $m\approx D$ and for $m\approx 0.9n$ when
$\b^*\!\!/2$ is accidentally close to $\a_+=D/2$.

The offline estimators (densely dashed lines), \gDirMX, \gSbX\ with
constant optimal parameters $\v\a^*$ and $\b^*$ mostly coincide
within $\pm 10$ nits with their variable $\vec{\v\a}^*$ and $\vbs$
online versions, except for very large $m$ they are slightly
better. This shows that making them online is essentially for free,
which is consistent with the close bounds for small $m$ in both
cases. This has been observed for other offline-online algorithm
pairs as well \cite{Hutter:05expertx}. There is very little gain in
knowing $\v\a^*$ or $\b^*$ in advance.
As expected off-line \gKTACLA\ significantly improves upon \gKTX\
for small $m$ and even beats \gSbXt\ by a couple of bits for
sufficiently small $m$, but breaks down for medium and large $m$,
and anyway is off-line.

These observations are rather consistent across uniform, Zipf, and
real data. Only for Zipf data, \gSAWBayes\ and \gKTACLA\ seem to be
worse, and the relative performance of many estimators on b\&w fax
{\em pic} is reversed.

The oracle estimators (dotted lines) possess significant extra
knowledge: \gKTAo\ the used alphabet $\A$, and \gLLto\ and \gHo\
even the counts $\v n$. The plots show the magnitude of this extra
knowledge.

\paradot{Summary}
Results are similar for other $(n,D,m)$ and $(n,D,\gamma)$
combinations but code length differences can be more or less
pronounced but are seldom reversed. In short, %
\gKTX\ performs very poorly unless $m\approx D$, and %
\gPerks\ and \gSSDC\ perform poorly unless $m\lesssim\ln n$; %
\gKTACLA, \gDirMX, \gSbX\ are not online; %
the oracles \gLLto, \gHo, \gKTAo\ are not realizable; %
and \gSAWBayes\ is extremely slow; %
which leaves \gDirMXt\ and \gSbXt\ as winners. %
They perform very similar unless $m$ gets very close to
$\min\{n,D\}$ in which case \gSbXt\ wins.

\section{List of Notation}\label{app:Notation}

\begin{tabbing}
  \hspace{0.13\textwidth} \= \hspace{0.73\textwidth} \= \kill
  {\bf Symbol }      \> {\bf Explanation}                                                    \\[0.5ex]
  $\X$               \> total (large) base alphabet of size $D$                             \\[0.5ex]
  $D=|\X|$           \> size of (large) base alphabet $\X$                                  \\[0.5ex]
  $n$                \> sequence length                                                      \\[0.5ex]
  $x_{1:n}$          \> total sequence                                                       \\[0.5ex]
  $n_i$              \> number of times $i$ appears in $x_{1:n}$                             \\[0.5ex]
  $\A\subseteq\X$    \> symbols actually appearing in sequence $x_{1:n}$                     \\[0.5ex]
  $m=|\A|$           \> size of alphabet used in $x_{1:n}$                                   \\[0.5ex]
  $i,j,k$            \> indices ranging over symbols in $\X$, $\A$, $\X\setminus\A$ respectively \\[0.5ex]
  $\nb:={n\over m},\nu:={m\over n}$ \> \hspace{5ex} average multiplicity of symbols and its inverse \\[0.5ex]
  $t$                \> current time ranging from $0$ to $n-1$                               \\[0.5ex]
  $x_{1:t}$          \> sequence seen so far                                                 \\[0.5ex]
  $\A_t$             \> $=\{x_1,...,x_t\}$ = symbols seen so far                             \\[0.5ex]
  $m_t=|\A_t|$       \> number of different symbols observed so far (in $x_{1:t}$)           \\[0.5ex]
  $x_{t+1}$          \> next symbol to be predicted                                          \\[0.5ex]
  $n^t_i$            \> number of times $i$ appears in $x_{1:t}$                             \\[0.5ex]
  $\newx$            \> set of $t$ for which $x_{t+1}$ is new, i.e.\ $x_{t+1}\not\in\A_t$    \\[0.5ex]
  $\oldx$            \> set of $t$ for which $x_{t+1}$ is old, i.e.\ $x_{t+1}\in\A_t$        \\[0.5ex]
  $P,Q$              \> probability over sequences                                           \\[0.5ex]
  $P^{\rm param}_{\rm name}$ \> parameterized and named probability                          \\[0.5ex]
  $R^{\rm param}_{\rm name}$ \> $= -\ln P^{\rm param}_{\rm name} - n\cdot H(\v n/n)$ = redundancy of $P^{\rm param}_{\rm name}$ \\[0.5ex]
  $\overline R,\underline R$      \> upper/lower bound on redundancy                         \\[0.5ex]
  $\CL$              \> code length in nits                                                  \\[0.5ex]
  $\t_i$             \> probability that $x_t=i$                                             \\[0.5ex]
  $\a_i,\a_+$        \> Dirichlet parameters and their sum                                   \\[0.5ex]
  $\b=\b_n,\b_t$     \> general (constant,variable) parameter $\b$                           \\[0.5ex]
  $\b^*\neq\b^*_n,\b^*_t$ \> optimal (constant,variable) parameter $\b$                      \\[0.5ex]
  $w^t_i$            \> weight of new symbol $i$ at time $t$                                 \\[0.5ex]
  $\ln$              \> Natural logaritm. Results are in `nits'                              \\[0.5ex]
  $\v v$             \> vector over alphabet $\X$                                            \\[0.5ex]
  $\vec v$           \> vector over time $t=0...n-1$                                         \\[0.5ex]
  $\G,\Psi$          \> Gamma and diGamma function                                           \\[0.5ex]
  $c$                \> constant $>0$ and $<\infty$                                          \\[-2ex]
\end{tabbing}

\end{document}